\documentclass[a4paper,12pt]{article}

\pdfoutput=1

\makeatletter
 
  \@addtoreset{equation}{section}
 \makeatother

\usepackage{amsfonts}
\usepackage{amssymb}
\usepackage{mathrsfs}
\usepackage{amsmath,amsthm}
\usepackage{amsmath}

\usepackage{setspace}

\usepackage{cite}

\usepackage{hyperref}

\usepackage[dvipdfmx]{graphicx}
\usepackage[dvipdfmx]{color}

\usepackage{here}

\usepackage{indentfirst}

\begin{document}

\newtheorem{theorem}{Theorem}
\newtheorem{lemma}{Lemma}
\newtheorem{proposition}{Proposition}
\newtheorem{cor}{Corollary}
\newtheorem{defn}{Definition}
\theoremstyle{definition}
\newtheorem{remark}{Remark}
\newtheorem{step}{Step}

\newcommand{\Cov}{\mathop {\rm Cov}}
\newcommand{\Var}{\mathop {\rm Var}}
\newcommand{\E}{\mathop {\rm E}}
\newcommand{\const }{\mathop {\rm const }}
\everymath {\displaystyle}

\newcommand{\ruby}[2]{
\leavevmode
\setbox0=\hbox{#1}
\setbox1=\hbox{\tiny #2}
\ifdim\wd0>\wd1 \dimen0=\wd0 \else \dimen0=\wd1 \fi
\hbox{
\kanjiskip=0pt plus 2fil
\xkanjiskip=0pt plus 2fil
\vbox{
\hbox to \dimen0{
\small \hfil#2\hfil}
\nointerlineskip
\hbox to \dimen0{\mathstrut\hfil#1\hfil}}}}

\def\qedsymbol{$\blacksquare$}

\renewcommand{\refname }{References}

\title{
Asymptotic Analysis for Spectral Risk Measures Parameterized by Confidence Level
}

\author{Takashi Kato\\
{\small Association of Mathematical Finance Laboratory (AMFiL)}\\
{\small               2--10, Kojimachi, Chiyoda, Tokyo 102--0083, Japan}\\
{\small \texttt{takashi.kato@mathfi-lab.com}}
}

\date{November 20, 2017}

\maketitle 

\begin{abstract}
We study the asymptotic behavior of the difference 
$\Delta \rho ^{X, Y}_\alpha := \rho _\alpha (X + Y) - \rho _\alpha (X)$ as $\alpha \rightarrow 1$, 
where $\rho_\alpha $ is a risk measure equipped with a confidence level parameter $0 < \alpha < 1$, 
and where $X$ and $Y$ are non-negative random variables whose tail probability functions are regularly varying. 
The case where $\rho _\alpha $ is the value-at-risk (VaR) at $\alpha $, 
is treated in \cite {Kato_IJTAF}. 
This paper investigates the case where $\rho _\alpha $ is a spectral risk measure 
that converges to the worst-case risk measure as $\alpha \rightarrow 1$. 
We give the asymptotic behavior of 
the difference between the marginal risk contribution and the Euler contribution of $Y$ to the portfolio $X + Y$. 
Similarly to \cite {Kato_IJTAF}, 
our results depend primarily on the relative magnitudes of the thicknesses of the tails  of $X$ and $Y$. 
We also conducted a numerical experiment, finding 
that when the tail of $X$ is sufficiently thicker than that of $Y$, 
$\Delta \rho ^{X, Y}_\alpha $ does not increase monotonically with $\alpha$ and 
takes a maximum at a confidence level strictly less than $1$. 
\end{abstract}

\everymath {\displaystyle}

\section{Introduction}

The purpose of this paper is to investigate the asymptotic behavior of the difference 
\begin{eqnarray}\label{Delta_rho}
\Delta \rho ^{X,Y}_\alpha := \rho _\alpha (X + Y) - \rho _\alpha (X)
\end{eqnarray}
as $\alpha \rightarrow 1$, where $X$ and $Y$ are fat-tailed random variables (loss variables) and 
$(\rho _\alpha )_{0 < \alpha < 1}$ is a family of risk measures. 
The case where $\rho _\alpha $ is an $\alpha$-percentile value-at-risk (VaR), 
has been treated in \cite {Kato_IJTAF}, where it was shown that
the asymptotic behavior of $\Delta \mathrm {VaR}^{X,Y}_\alpha $ drastically changes 
according to the relative magnitudes of the thicknesses of the tails of $X$ and $Y$ 
(the definition of the VaR is given in (\ref {def_VaR}) in the next section). 
In this paper, we study a  progressive case in which 
$\rho _\alpha $ is given as a parameterized spectral risk measure, 
and we obtain similar results as in \cite {Kato_IJTAF}. 
In particular, we find that if $X$ and $Y$ are independent and if the tail of $X$ is sufficiently fatter than that of $Y$, then
$\Delta \rho ^{X,Y}_\alpha $ converges to the expected value $\E [Y]$ as $\alpha \rightarrow 1$ 
whenever $(\rho _\alpha )_{0 < \alpha < 1}$ are spectral risk measures converging to 
a risk measure of the worst case scenario. That is, whenever
\begin{eqnarray}\label{conv_worst_case}
\rho _\alpha (Z) \ \mathop {\longrightarrow }_{\alpha \rightarrow 1} \ \mathop {\rm ess\;sup}_{\omega }Z(\omega ) 
\end{eqnarray}
for each loss random variable $Z$ in some sense. 
Our result does not require any specific form for $\rho _\alpha $, 
implying that this property is robust. 
Furthermore, assuming some technical conditions for the probability density functions of $X$ and $Y$, 
we study the asymptotic behavior of the Euler contribution, defined as 
\begin{eqnarray}\label{EC}
\rho ^\mathrm {Euler}_\alpha (Y | X+Y) = \frac{\partial }{\partial h}\rho _\alpha (X + hY)\Big |_{h = 1}
\end{eqnarray}
(see Remark 17.1 in \cite {Tasche2008}), 
and show that $\Delta \rho ^{X, Y}_\alpha $ is asymptotically equivalent to $\delta \rho ^\mathrm {Euler}_\alpha (Y | X+Y)$ 
as $\alpha \rightarrow 1$. 
Here, $\delta \in (0, 1]$ is a constant 
determined according to the relative magnitudes of the thicknesses of the tails of $X$ and $Y$. 

We now briefly review the financial background for this study. 
In quantitative financial risk management, 
it is important to capture tail loss events by using adequate risk measures. 
One of the most standard risk measures is the VaR. 
The Basel Accords, which provide a set of recommendations for regulations in the banking industry, 
essentially recommend using VaR as a measure of risk capital for banks. 
VaRs are indeed simple, useful, and 
their values are easy to interpret. 
For instance, a yearly $99.9\%$ VaR calculated as $x_0$
means that the probability of a risk event 
with a realized loss larger than $x_0$ is $0.1\%$. 
In other words, an amount $x_0$ of risk capital is sufficient
to prevent a default with 99.9\% probability. 
The meaning of the amount $x_0$ is therefore easy to understand. 
However, VaRs are often criticized for their lack of subadditivity 
(see, for instance, \cite {Acerbi, Acerbi-Tasche1, Embrechts}). 
VaRs do not reflect the risk diversification effect. 

The expected shortfall (ES) has been proposed 
as an alternative risk measure that is coherent (in particular, subadditive) and tractable, with the risk amount at least that of the corresponding VaR. 
Note that there are various versions of ES, such as 
the conditional value-at-risk (CVaR), 
the average value-at-risk (AVaR), 
the tail conditional expectation (TCE), 
and the worst conditional expectation (WCE). 
These are all equivalent under some natural assumptions 
(see \cite {Acerbi-Nordio-Sirtori, Acerbi-Tasche1, Acerbi-Tasche2, Artzner-Delbaen-Eber-Heath}). 
It should be noted that the Basel Accords have also considered recently the adoption of ESs as a minimal capital requirement, in order to better capture market tail risks (see for instance \cite {BCBS2012, BCBS2016}). 

A spectral risk measure (SRM) has been proposed as a generalization of ESs, 
in \cite {Acerbi}. 
SRMs are characterized by a weight function $\phi$ that represents 
the significance of each confidence level for the risk manager. 
SRMs are equivalent to comonotonic law-invariant coherent risk measures 
(see Remark \ref {rem_SRM} in the next section). 

VaRs and ESs as risk measures depend on a confidence level parameter $\alpha \in (0, 1)$. 
We let $\mathrm {VaR}_\alpha $ (resp., $\mathrm {ES}_\alpha $) denote the VaR (resp., ES) with confidence level $\alpha $. 
When $\alpha $ is close to $1$, the values of $\mathrm {VaR}_\alpha $ and $\mathrm {ES}_\alpha $ are 
increasing without bound as in (\ref {conv_worst_case}). 
The parameter $\alpha $ corresponds to the risk aversion level of the risk manager. 
Higher values of $\alpha $ indicate that 
the risk manager is more risk-averse and evaluates the tail risk as more severe. 

In this paper, we consider a family $(\rho _\alpha )_{0 < \alpha < 1}$ of 
SRMs parameterized by the confidence level $\alpha $. 
we make a mathematical assumption that intuitively implies situation (\ref {conv_worst_case}) and investigate the asymptotic behaviors of (\ref {Delta_rho}) and (\ref {EC}) 
as $\alpha \rightarrow 1$, 
when the tail probability function of $X$ (resp., $Y$) is regularly varying with index $-\beta $ (resp., $-\gamma $). 
Our main theorem asserts that the asymptotic behaviors of (\ref {Delta_rho}) and (\ref {EC}) 
strongly depend on 
the relative magnitudes of $\beta $ and $\gamma$. 
Note that our results include 
the case $\rho _\alpha = \mathrm {ES}_\alpha $, 
the inclusion of which was discussed as a future task in \cite {Kato_IJTAF}. 

The rest of this paper is organized as follows. 
In Section \ref {sec_pre}, we prepare the basic settings and introduce the definitions for SRMs based on confidence level. 
In Section \ref {sec_main}, we give our main results. 
We numerically verify our results in Section \ref {sec_numerical}. 
Finally, Section \ref {sec_conclusion} summarizes our studies. 
Throughout the main part of this paper, we assume that $X$ and $Y$ are independent. The more general case where $X$ and $Y$ are not independent is studied in Appendix \ref {sec_add}. 
All proofs are given in Appendix \ref {sec_proofs}.

\section{Preliminaries}\label{sec_pre}

Let $(\Omega , \mathcal {F}, P)$ be a standard probability space and let
$\mathscr {L}_+$ denote a set of non-negative random variables defined on $(\Omega , \mathcal {F}, P)$. 
For each $Z\in \mathscr {L}_+$, we denote by $F_Z$ the distribution function of $Z$ 
and by $\bar{F}_Z$ its tail probability function; that is, 
$F_Z(z) = P(Z\leq z)$ and $\bar{F}_Z(z) = P(Z > z)$. 
Moreover, for each $\alpha \in (0, 1)$, we define 
\begin{eqnarray}\label{def_VaR}
\mathrm {VaR}_\alpha (Z) = \inf \{ z\in \mathbb {R}\ ; \ P(Z \leq z) \geq \alpha  \} . 
\end{eqnarray}
Note that $\mathrm {VaR}_\alpha (Z)$ is exactly the left-continuous version 
of the generalized inverse function of $F_Z$. 

We now introduce the definition of SRMs. 

\begin{defn} \ 
\begin{itemize}
 \item [{\rm (i)}]A Borel measurable function $\phi: [0, 1) \longrightarrow [0, \infty )$ is called an admissible spectrum 
if $\phi $ is right-continuous, non-decreasing, and satisfies 
\begin{eqnarray}\label{cond_normalization}
\int ^1_0\phi (\alpha )d\alpha  = 1. 
\end{eqnarray} 
 \item [{\rm (ii)}] A risk measure $\rho: \mathscr {L}_+ \longrightarrow [0, \infty )$ is called an SRM 
if there is an admissible spectrum $\phi $ such that $\rho = M_\phi $, where 
\begin{eqnarray*}
M_\phi (Z) = \int ^1_0\mathrm {VaR}_\alpha (Z)\phi (\alpha )d\alpha , \ \ Z\in \mathscr {L}_+. 
\end{eqnarray*}
\end{itemize}
\end{defn}

\begin{remark}\label{rem_SRM}
SRMs are law-invariant, comonotonic, and coherent risk measures. 
However, as shown in \cite{Foellmer-Schied, Kusuoka, Jouini-Schachermayer-Touzi}, if $(\Omega , \mathcal {F}, P)$ is atomless, 
then 
for any law-invariant comonotonic convex risk measure $\rho $, there is a probability measure $\mu $ on $[0, 1]$ such that 
\begin{eqnarray}\label{Kusuoka-rep}
\rho (Z) = \int ^1_0\mathrm {ES}_\alpha (Z)\mu (d\alpha ), 
\end{eqnarray}
for each $Z\in L^\infty (\Omega , \mathcal {F}, P)$. This is due to the generalized Kusuoka representation theorem (Theorem 4.93 in \cite {Foellmer-Schied}), 
where $\mathrm {ES}_\alpha (Z)$ is the $\alpha $-percentile expected shortfall of $Z$: 
\begin{eqnarray}\label{def_ES}
\mathrm {ES}_\alpha (Z) = \frac{1}{1 - \alpha }\int ^1_\alpha \mathrm {VaR}_u(Z)du. 
\end{eqnarray}
Moreover, such a $\rho $ is always coherent and satisfies the Fatou property \cite {Jouini-Schachermayer-Touzi}. 
Furthermore, representation (\ref {Kusuoka-rep}) can also be rewritten as $\rho (Z) = M_{\phi _\mu }(Z)$, where 
\begin{eqnarray*}
\phi _\mu (\alpha ) = \int ^1_0\frac{1}{1 - u}1_{[0, \alpha ]}(u)\mu (du). 
\end{eqnarray*}
Here, it is easy to see that $\phi _\mu $ is non-negative, non-decreasing, right-continunous, and satisfies
\begin{eqnarray*}
\int ^1_0\phi _\mu (\alpha )d\alpha = \int ^1_0\frac{1}{1-u}\int ^1_01_{\{ 0 \leq u \leq  \alpha  \}}d\alpha \mu (du) = 1, 
\end{eqnarray*}
meaning that $\phi _\mu $ is an admissible spectrum (see \cite {Shapiro}). 
Therefore, any law-invariant comonotonic convex (or coherent) risk measure is completely 
characterized as an SRM. 
Arguments similar to those above, replacing $L^\infty (\Omega , \mathcal {F}, P)$ with 
$L^p(\Omega , \mathcal {F}, P)$, where $1\leq p < \infty $, 
can be found in \cite {Pflug-Roemisch, Shapiro}. 
\end{remark}

Next, we introduce a family $(\rho _\alpha )_{0 < \alpha < 1}$ of SRMs parameterized by the confidence level $\alpha $. 

\begin{defn} Let $(\phi _\alpha )_{0 < \alpha < 1}$ be a familly of admissible spectra and let $\rho _\alpha = M_{\phi _\alpha }$. 
Then $(\rho _\alpha )_{0 < \alpha < 1}$ is called a set of confidence-level-based spectral risk measures (CLBSRMs) if 
\begin{eqnarray}\label{def_CLBSRM}
\Phi _\alpha \ \longrightarrow ^\mathrm {w} \ \delta _1, \ \ \alpha \rightarrow 1, 
\end{eqnarray}
where $\Phi _\alpha $ is a probability measure on $[0, 1]$ defined by $\Phi _\alpha (du) = \phi _\alpha (u)du$ and 
$\delta _1$ is the Dirac measure with unit mass at $1$. 
\end{defn}

Condition (\ref {def_CLBSRM}) formally implies (\ref {conv_worst_case}). 
Indeed, if $Z\in \mathscr {L}_+$ is a bounded random variable with a distribution function that is continuous and strictly increasing on $[0, z^*]$, 
where $z^* = \mathop {\rm esssup}_{\omega }Z(\omega )$, 
then the function $u\mapsto \mathrm {VaR}_u(Z)$ is bounded and continuous, so that (\ref {def_CLBSRM}) gives 
\begin{eqnarray*}
\rho _\alpha (Z) = \int ^1_0\mathrm {VaR}_u(Z)\Phi _\alpha (du) \longrightarrow z^*, \ \ \alpha \rightarrow 1, 
\end{eqnarray*}
where we recognize $\mathrm {VaR}_1(Z) = F^{-1}_Z(1) = z^*$. 
Moreover, we see that 
\begin{lemma}\label{lem_pointwise}
Relation $(\ref {def_CLBSRM})$ is equivalent to 
\begin{eqnarray}\label{conv_pointwise}
\phi _\alpha (u) \longrightarrow  0, \ \ \alpha \rightarrow 1 \ \mbox { for each } \ u\in [0, 1). 
\end{eqnarray}
\end{lemma}

We now give some examples of CLBSRMs. \\

\noindent 
{\it Example 1.} Expected Shortfalls 

$(\mathrm {ES}_\alpha )_{0 < \alpha < 1}$ defined by (\ref {def_ES}) is a typical example of a CLBSRM. 
The corresponding admissible spectra are given as 
\begin{eqnarray*}
\phi ^{\mathrm {ES}}_\alpha (u) = \frac{1}{1-\alpha }1_{[\alpha  , 1)}(u). 
\end{eqnarray*}
It is easy to see that (\ref {def_CLBSRM}) does hold. 
Indeed, for any bounded continuous function $f$ defined on $[0, 1]$, we see that 
\begin{eqnarray*}
\frac{1}{1 - \alpha }\int ^1_\alpha f(u)du = 
\int ^1_0f(u + \alpha (1 - u))du \longrightarrow f(1), \ \ \alpha \rightarrow 1
\end{eqnarray*}
due to the bounded convergence theorem. 
Equivalently, we can also check that $(\mathrm {ES}_\alpha )_\alpha $ satisfies (\ref {conv_pointwise}). 

$\mathrm {ES}_\alpha $ is characterized as the smallest law-invariant coherent risk measures that are greater than or equal to $\mathrm {VaR}_\alpha $ \cite {Kusuoka}. 
Note that if the distribution function of the target random variable $Z$ is continuous, then
$\mathrm {ES}_\alpha (Z)$ coincides with $\mathrm {CVaR}_\alpha (Z)$, where 
\begin{eqnarray*}
\mathrm {CVaR}_\alpha (Z) = \E [Z\ | \ Z\geq \mathrm {VaR}_\alpha (Z)] 
\end{eqnarray*}
(see \cite {Acerbi-Tasche2} for details). \\

\noindent 
{\it Example 2.} Exponential/Power SRMs 

An admissible spectrum $\phi $ corresponding to an SRM $M_\phi $ represents the
preferences of a risk manager for each quantile of the loss distribution. 
Therefore, the form taken by $\phi $ corresponds to the manager's risk aversion, 
which is also described in terms of utility functions in classical decision theory. 
Recently, the relation between expected utility functions and SRMs has been studied, 
though it has not been entirely resolved. 
Here we introduce some examples of SRMs based on specific utility functions. 

The exponential utility function is a typical example of tractable utility functions 
\begin{eqnarray*}
U_\gamma (p) = -\frac{e^{-\gamma p}}{\gamma }, 
\end{eqnarray*}
where $p$ denotes the 
profit-and-loss ($p > 0$ indicating profit) 
and $\gamma $ characterizes the degree of risk preference. 
We focus on the case $0 < \gamma < \infty $ so that $U_\gamma $ describes a risk-averse utility function. 
We transform the parameter $\gamma $ into the confidence level $\alpha \in (0, 1)$ using 
$\alpha = (2/\pi )\tan ^{-1}\gamma $. 
Note that the original parameter $\gamma $ can be recovered using the inverse 
$\gamma = t_\alpha := \tan (\pi \alpha /2)$. 
The exponential utility of the loss $l$ with confidence level $\alpha $ is then given as 
$U_{t_\alpha }(-l) = -e^{lt_\alpha }/t_\alpha $. 
Cotter and Dowd \cite {Cotter-Dowd} have proposed 
an SRM $\rho ^{\mathrm {EXP}}_\alpha = M_{\phi ^{\mathrm {EXP}}_\alpha }$ based on the exponential utility by constructing 
an admissible spectrum 
$\phi ^{\mathrm {EXP}}_\alpha (u) = -\lambda U_{t_\alpha }(-u)$ 
for some $\lambda > 0$, so that 
$\phi ^{\mathrm {EXP}}_\alpha (u)$ satisfies (\ref {cond_normalization}). 
Then, $\lambda $ must be set as 
$t_\alpha ^2/(e^{t_\alpha } - 1)$, 
giving 
\begin{eqnarray*}
\phi ^{\mathrm {EXP}}_\alpha (u) = 
\frac{t_\alpha e^{-t_\alpha (1-u)}}{1-e^{-t_\alpha }}. 
\end{eqnarray*}
Note that the theoretical validity of the above method is still unclear. 
Other methods to adequately construct SRMs from exponential utility functions 
have been discussed in \cite{Brandtner-Kuersten, Sriboonchitta-Nguyen-Kreinovich, Waechter-Mazzoni}, 
but no definite answer has been reached. 
In particular, it is pointed out in \cite {Brandtner-Kuersten} that there exists
no general consistency between expected utility theory and SRM-decision making. 
In any case, we can easily verify that $(\phi ^\mathrm {EXP}_\alpha )_\alpha $ as defined above satisfies 
(\ref {def_CLBSRM})--(\ref {conv_pointwise}), which implies that 
$(\rho ^\mathrm {EXP}_\alpha )_\alpha $ is actually a CLBSRM. 

Similarly to the above, an SRM $\rho ^\mathrm {POW}_\alpha = M_{\phi ^\mathrm {POW}_\alpha }$ based on the power utility function has been studied in \cite {Dowd-Cotter-Sorwar}. 
After changing the risk aversion parameter to the confidence level $\alpha \in (0, 1)$ as above, 
$\phi ^\mathrm {POW}_\alpha $ is given as 
\begin{eqnarray*}
\phi ^\mathrm {POW}_\alpha (u) = \frac{u^{\alpha /(1-\alpha )}}{1-\alpha }. 
\end{eqnarray*}
We can also verify that $(\rho ^\mathrm {POW}_\alpha )_{0 < \alpha < 1}$ is a CLBSRM. \vspace{3mm}

We now introduce some notations and definitions used in asymptotic analysis and extreme value theory. 

Let $f$ and $g$ be positive functions defined on $[x_0, x_1)$, where $x_0\in [0, \infty )$ and $x_1 \in (x_0, \infty ]$. 
We say that $f$ and $g$ are asymptotically equivalent (denoted as $f\sim g$) as $x\rightarrow x_1$ 
if 
$\lim _{x\rightarrow x_1}f(x)/g(x) = 1$. 
When $x_1 = \infty $, 
we say that $f$ is regularly varying with index $k\in \mathbb {R}$ if 
it holds that $\lim _{x\rightarrow \infty }f(tx)/f(x) = t^k$ for each $t > 0$. 
Moreover, we say that $f$ is ultimately decreasing if 
$f$ is non-increasing on $[x_2, \infty )$ for some $x_2 > 0$. 
For more details, we refer the reader to \cite {Bingham-Goldie-Teugels, Embrechts-Klueppelberg-Mikosch}.

\section{Main results}\label{sec_main}

Our main purpose is 
to investigate the property of (\ref {Delta_rho}) for a CLBSRM $(\rho _\alpha )_{0 < \alpha < 1}$ and 
random variables $X, Y\in \mathscr {L}_+$ whose 
distributions are fat-tailed. 
To consider this case, 
we assume that $\bar{F}_X$ and $\bar{F}_Y$ are regularly varying functions with indices $-\beta $ and $-\gamma $, respectively. 
That is, $\bar{F}_X(x), \bar{F}_Y(x) > 0$ for each $x\geq 0$ and 
\begin{eqnarray}\label{regular_variation}
\lim _{x\rightarrow \infty }\frac{\bar{F}_X(tx)}{\bar{F}_X(x)} = t^{-\beta }, \ \ 
\lim _{x\rightarrow \infty }\frac{\bar{F}_Y(tx)}{\bar{F}_Y(x)} = t^{-\gamma }, \ \ t > 0 
\end{eqnarray}
for some $\beta , \gamma > 0$. 

In \cite {Kato_IJTAF}, 
we study the asymptotic property of (\ref {Delta_rho}) as $\alpha \rightarrow 1$ when 
$\rho _\alpha = \mathrm {VaR}_\alpha $. 
The results display
the following five patterns: 
(i) $\beta + 1 < \gamma $, (ii) $\beta < \gamma \leq \beta + 1$, 
(iii) $\beta = \gamma $, 
(iv) $\gamma < \beta \leq \gamma + 1$, and 
(v) $\gamma + 1 < \beta $. 
In cases (iv) and (v), we consider the difference $\Delta \mathrm {VaR}^{Y, X}_\alpha $ instead of 
$\Delta \mathrm {VaR}^{X, Y}_\alpha $, and the results are restated consequences of cases (i) and (ii). 
Hence, we assume here that $\beta \leq \gamma $ and focus on cases (i)--(iii) only. 
We further assume that $\beta > 1$. 
This assumption guarantees the integrability of $X$ and $Y$ 
(see, for instance, Proposition A3.8 in \cite {Embrechts-Klueppelberg-Mikosch}). 

Let $(\rho _\alpha )_{0 < \alpha < 1}$ be a CLBSRM with a family of admissible spectra $(\phi _\alpha )_{0 < \alpha < 1}$. 
Here we assume that 
\begin{eqnarray}\label{ass_bdd_phi}
\phi _{\alpha }(1-) = \lim _{u\rightarrow 1}\phi _\alpha (u) < \infty 
\end{eqnarray}
for each $\alpha \in (0, 1)$. 
Then, Lemma A.23 in \cite {Foellmer-Schied} implies that 
\begin{eqnarray*}
\rho _\alpha (X + Y) \leq \phi _\alpha (1-)\int ^1_0\mathrm {VaR}_u(X+Y)du = 
\phi _{\alpha }(1-)(\E [X] + \E [Y]) < \infty 
\end{eqnarray*}
for each $0 < \alpha < 1$. 
This immediately implies that $\rho _\alpha (X), \rho _\alpha (Y) < \infty $. 
Furthermore, by (17.9b) and Proposition 17.2 in \cite {Tasche2008}, we see that 
\begin{eqnarray}\label{ineq_general}
\Delta \rho ^{X, Y}_\alpha \leq \rho ^\mathrm {Euler}_\alpha (Y | X+Y) \leq \rho _\alpha (Y), 
\end{eqnarray}
where $\rho ^\mathrm {Euler}(Y | X+Y)$ is given by (\ref {EC}) if 
$\rho _\alpha (X + hY)$ is continuously differentiable in $h$.
Note that inequality (\ref {ineq_general}) holds for each $0 < \alpha < 1$ whenever $\rho _\alpha $ is coherent. 

Our main purpose in this section is to investigate in detail the asymptotic behavior of $\Delta \rho ^{X, Y}_\alpha $, as well as $\rho ^\mathrm {Euler}_\alpha (Y | X + Y)$ if it is defined, as $\alpha \rightarrow 1$. 
To clearly state our main results, we establish the following conditions, 
which are assumed to hold in Section 4 of \cite {Kato_IJTAF}. 
\begin{itemize}
 \item[ \mbox{[C1]}] \ $X$ and $Y$ are independent. 
 \item[ \mbox{[C2]}] \ There is some $x_0 \geq  0$ such that 
$F_X$ has a positive, non-increasing density function $f_X$ on $[x_0, \infty )$; that is,
$F_X(x) = F_X(x_0) + \int ^x_{x_0}f_X(y)dy, \ \ x \geq x_0$. 
 \item[ \mbox{[C3]}] \ The function $x^{\gamma - \beta }\bar{F}_Y(x)/\bar{F}_X(x)$ converges to 
some real number $k$ as $x\rightarrow \infty $. 
\end{itemize}

Let us adopt the notation 
\begin{eqnarray}\label{def_M_bar}
\bar{M}(\alpha ) = 
\left\{
\begin{array}{ll}
 	\E [Y] & 	\mbox { if } \beta + 1 < \gamma , \\
 	\frac{k}{\beta }\int ^1_0\mathrm {VaR}_u(X)^{\beta + 1 - \gamma }\phi _\alpha (u)du& \mbox { if } \beta < \gamma \leq \beta + 1, \\
 	\{ (1 + k)^{1/\beta } - 1 \}\rho _\alpha (X)& \mbox { if } \beta = \gamma 
\end{array}
\right. 
\end{eqnarray}
for $0 < \alpha < 1$. 
Note that $\bar{M}(\alpha )$ is finite for each fixed $\alpha \in (0, 1)$ 
(see Corollary \ref {cor_finiteness_M} in Appendix \ref {sec_proofs}). 
Our main results are the two following theorems. 

\begin{theorem}\label{th_main} \ Assuming {\rm [C1]--[C3]}, 
$\Delta \rho ^{X, Y}_\alpha \sim \bar{M}(\alpha )$ as $\alpha \rightarrow 1$. 
\end{theorem}

Formally, assertions (i)--(iii) of Theorem 4.1 in \cite {Kato_IJTAF} 
are the same as the assumptions of Theorem \ref {th_main}, 
by setting $\Phi _\alpha = \delta _{\alpha }$. 
That is, we have $\Delta \mathrm {VaR}^{X, Y}_\alpha \sim \bar{f}(\alpha )$ as $\alpha \rightarrow 1$, 
where 
\begin{eqnarray}\label{def_f_bar}
\bar{f}(\alpha ) = 
\left\{
\begin{array}{ll}
 	\E [Y] & 	\mbox { if } \beta + 1 < \gamma , \\
 	\frac{k}{\beta }\mathrm {VaR}_\alpha (X)^{\beta + 1 - \gamma }& \mbox { if } \beta < \gamma \leq \beta + 1, \\
 	\{ (1 + k)^{1/\beta } - 1 \}\mathrm {VaR}_\alpha (X)& \mbox { if } \beta = \gamma . 
\end{array}
\right. 
\end{eqnarray}
Theorem \ref {th_main} justifies the following relation: 
\begin{eqnarray*}
\Delta \rho ^{X, Y}_\alpha = \int ^1_0\Delta \mathrm {VaR}^{X, Y}_u\phi _\alpha (u)du 
\sim 
\int ^1_0\bar{f}(u)\phi _\alpha (u)du = \bar{M}(\alpha ), \ \ \alpha \rightarrow 1. 
\end{eqnarray*}

Note that condition [C3] is not required for Theorem \ref {th_main} when $\beta + 1 < \gamma $. 
Moreover, when $\beta + 1 < \gamma $, 
Theorem \ref {th_main} implies that 
$\Delta \rho ^{X, Y}_\alpha $ converges to $\E [Y]$ as $\alpha \rightarrow 1$. 
The limit $\E [Y]$ does not depend on the forms of $(\phi _\alpha )_\alpha $, 
so this result is robust. 
The second main result is as follows. 

\begin{theorem}\label{th_equivalence_Euler_EL}
Assume {\rm {[C1]}} and {\rm {[C3]}}. Moreover, 
assume that  
\begin{itemize}
 \item [{\rm [C4]}]$X$ and $Y$ have positive, continuous, and ultimately decreasing density functions $f_X$ and $f_Y$, respectively, on $[0, \infty )$. 
\end{itemize}
Under these assumptions, $\rho ^\mathrm {Euler}_\alpha (Y | X + Y) \sim \bar{M}(\alpha )/\delta $ as $\alpha \rightarrow 1$, where 
$\delta $ is a positive constant given by 
\begin{eqnarray}\label{def_delta}
\delta  = 
\left\{
\begin{array}{cl}
 	1 & 	\mbox { \rm if } \beta + 1 < \gamma , \\
 	k / (\E [Y]\beta  + k\gamma )& \mbox { \rm if } \beta + 1 = \gamma , \\
 	1 / \gamma & \mbox { \rm if } \beta < \gamma < \beta + 1, \\
 	\{ 1 + k - (1+k)^{1-1/\beta }\} / k& \mbox { \rm if } \beta = \gamma . 
\end{array}
\right. 
\end{eqnarray}
\end{theorem}

Theorems \ref {th_main} and \ref {th_equivalence_Euler_EL} together imply that 
if $X$ and $Y$ are independent, and if 
$F_X$ and $F_Y$ have adequate density functions, 
then 
\begin{eqnarray}\label{asy_behavior_marginal_Euler}
\Delta \rho ^{X, Y}_\alpha \sim \delta \rho ^\mathrm {Euler}_\alpha (Y | X + Y), \ \ \alpha \rightarrow 1. 
\end{eqnarray}
Note that $\delta $ is always smaller than or equal to $1$, 
so that (\ref {asy_behavior_marginal_Euler}) is consistent with inequality (\ref {ineq_general}). 
In particular, if $\beta + 1 < \gamma $, then the asymptotic equivalence between 
the marginal risk contribution $\Delta \rho ^{X, Y}_\alpha $ 
and 
the Euler contribution $\rho ^\mathrm {Euler}_\alpha (Y | X + Y)$ is justified (see (17.10) in \cite {Tasche2008}). 

Note that $\Delta \rho _\alpha ^{X, Y}$ is always larger than or equal to $\E [Y]$ 
so long as the random vector $(X, Y)$ satisfies a suitable 
technical condition, such as Assumption (S) in \cite {Tasche2000}. 
(Here, we modify some conditions of the original version of Assumption (S) 
to facilitate focusing on non-negative random variables.) 
Indeed, because $\rho _\alpha $ is a convex risk measure, 
the function $r(h):= \rho _\alpha (X + hY)$ is convex. 
Thus, we get 
\begin{eqnarray}\label{lower_bound}
\Delta \rho _\alpha ^{X, Y} = r(1) - r(0) \geq r'(0) = \E [Y], 
\end{eqnarray}
where the last equality in the above relation is obtained from (see (5.12) in \cite {Tasche2000})
\begin{eqnarray}\label{derivative_VaR}
\frac{\partial }{\partial h}\mathrm {VaR}_\alpha (X+hY) = \E [Y | X+hY = \mathrm {VaR}_\alpha (X+hY)]
\end{eqnarray}
and 
\begin{eqnarray*}
\frac{\partial }{\partial h}\Big |_{h=0}\rho _\alpha (X + hY) = 
\int ^1_0\frac{\partial }{\partial h}\Big |_{h=0}\mathrm {VaR}_u(X + hY)\phi _\alpha (u)du = \E [Y]
\end{eqnarray*}
due to the dominated convergence theorem. 
Therefore, 
if $\beta + 1 < \gamma $, then
\begin{eqnarray*}
\E [Y] \leq \Delta \rho ^{X, Y}_\alpha \sim \rho ^\mathrm {Euler}_\alpha ( Y | X + Y) \longrightarrow \E [Y], \ \ \alpha \rightarrow 1. 
\end{eqnarray*}
In Section \ref {sec_numerical}, 
we numerically verify the above relation. 
Note that we can also verify a version of Assumption (S) under [C4].

\begin{remark}\label{rem_main} \ 
\begin{itemize}
 \item [(i)] If $F_X$ is continuous, then  
$F_X(X)$ has a uniform distribution on $(0, 1)$ 
(see, for instance, Lemma A.21 in \cite {Foellmer-Schied}). 
Therefore, $\bar{M}(\alpha )$ with $\beta < \gamma \leq \beta + 1$ is rewritten as 
\begin{eqnarray*}
\bar{M}(\alpha ) = \frac{k}{\beta }\E \hspace{0mm}^{Q^X_\alpha }[X^{\beta + 1 - \gamma }], 
\end{eqnarray*}
where $\E \hspace{0mm}^{Q^X_\alpha }$ denotes the expectation operator 
with respect to the probability measure $Q^X_\alpha $ defined as 
\begin{eqnarray}\label{def_Q}
\frac{dQ^X_\alpha }{dP} = \phi _\alpha (F_X(X)). 
\end{eqnarray}
Note that we have 
$\rho _\alpha (X) = \E \hspace{0mm}^{Q^X_\alpha }[X]$, 
and so $Q^X_\alpha $ represents the risk scenario that attains the maximum in the following 
robust representation of $\rho _\alpha (X)$: 
\begin{eqnarray*}
\rho _\alpha (X) = \max _{Q\in \mathscr {Q}}\E \hspace{0mm}^Q[X], 
\end{eqnarray*}
where $\mathscr {Q}$ is a set of probability measures on $(\Omega , \mathcal {F})$. 
Also note that if $\rho _\alpha = \mathrm {ES}_\alpha $, then $Q^X_\alpha $ is given by 
\begin{eqnarray*}
\frac{dQ^X_\alpha }{dP} = \frac{1}{1 - \alpha }1_{\{ X \geq \mathrm {VaR}_\alpha (X) \}}, 
\end{eqnarray*}
and therefore
\begin{eqnarray*}
\E \hspace{0mm}^{Q^X_\alpha }[X^{\beta + 1 - \gamma }] = 
\E [X^{\beta + 1 - \gamma } | X\geq \mathrm {VaR}_\alpha (X)]. 
\end{eqnarray*}
Until the end of Remark \ref {rem_main}, 
we assume that $F_X$ and $F_Y$ are continuous. 
 \item [(ii)]In fact, we can relax the independence condition [C1] so that 
$X$ may weakly depend on $Y$ within the negligible joint tail condition 
(see Remark A.1 in \cite {Kato_IJTAF}). 
In this case, under some additional assumptions such as [A5] and [A6] in \cite {Kato_IJTAF}, 
we can make the same assertion as in Theorem \ref {th_main}, 
where the value $\E [Y]$ in the definition (\ref {def_M_bar}) of $\bar{M}(\alpha )$ is replaced 
by $\E \hspace{0mm}^{Q^X_\alpha }[Y]$. 
In particular, if $\beta + 1 < \gamma $,
then 
\begin{eqnarray}\label{case_weak_dependent}
\Delta \rho ^{X, Y}_\alpha \sim 
\E \hspace{0mm}^{Q^X_\alpha }[Y], \ \ \alpha \rightarrow 1. 
\end{eqnarray}
Indeed, our proof in Appendix \ref {sec_proofs} also works 
by applying Theorem A.1 in \cite {Kato_IJTAF} instead of Theorem 4.1. 
Note that we need some additional condition to have that 
\begin{eqnarray}\label{cond_liminf}
\liminf _{\alpha \rightarrow \infty }\E \hspace{0mm}^{Q^X_\alpha }[Y] > 0 
\end{eqnarray}
(see Proposition \ref {prop_lower_bdd} in Appendix \ref {sec_proofs}). 
 \item [(iii)] As mentioned in Appendix A.1 of \cite {Kato_IJTAF}, 
we can get another version of Theorem A.1 by switching the roles of $X+Y$ and $X$ and 
by imposing modified (though somewhat artificial) mathematical conditions such as [A5'] and [A6'] in \cite {Kato_IJTAF}. 
In particular, if $\beta + 1 < \gamma $, we see that 
\begin{eqnarray}\label {asy_eq_VaR}
\Delta \mathrm {VaR}_\alpha ^{X, Y} \sim 
\E [Y | X + Y = \mathrm {VaR}_\alpha (X + Y)], \ \ \alpha \rightarrow 1
\end{eqnarray}
and then (by the same proof as Theorem \ref {th_main} with (\ref {asy_eq_VaR}))
\begin{eqnarray}\label{asy_eq_contribution}
\Delta \rho _\alpha ^{X, Y} \sim 
\E \hspace{0mm}^{Q^{X+Y}_\alpha }[Y] = \rho ^\mathrm {Euler}_\alpha (Y | X + Y), \ \ \alpha \rightarrow 1
\end{eqnarray}
under some assumptions. 
Here, $Q^{X+Y}_\alpha $ is a probability measure defined by $(\ref {def_Q})$ 
with replacing $X$ by $X+Y$. 
If $X$ and $Y$ are independent (with natural assumptions on the density functions), then
(\ref {asy_behavior_marginal_Euler}) implies that (\ref {asy_eq_contribution}) is also true. 
Here, note that the last equality of (\ref {asy_eq_contribution}) is obtained by 
(\ref {EC}), (\ref {derivative_VaR}), and the dominated convergence theorem. 
Indeed, we have 
\begin{eqnarray}\label{Euler_rho}
\rho ^\mathrm {Euler}_\alpha (Y | X + Y) = 
\int ^1_0\E [Y | X + Y = \mathrm {VaR}_u(X + Y)]\phi _\alpha (u)du = 
\E \hspace{0mm}^{Q^{X+Y}_\alpha }[Y] 
\end{eqnarray}
because $F_{X+Y}(X+Y)$ is uniformly distributed on $(0, 1)$. 
In Appendix \ref {sec_add}, 
we will show that 
under some technical conditions that are more natural than 
both [A5]--[A6] and [A5']--[A6'] in \cite {Kato_IJTAF}, 
relations (\ref {case_weak_dependent}) and (\ref {asy_eq_contribution}) simultaneously hold 
in the case $\beta + 1 < \gamma $, even if $X$ and $Y$ are dependent.

Note that if $\rho _\alpha = \mathrm {ES}_\alpha $, then 
\begin{eqnarray*}
\E \hspace{0mm}^{Q^{X+Y}_\alpha }[Y] = \mathrm {ES}^\mathrm {Euler}_\alpha (Y | X + Y) = 
\E [Y | X + Y \geq \mathrm {VaR}_\alpha (X + Y)], 
\end{eqnarray*}
which is known as the component CVaR (also known as the CVaR contribution) and widely used, particularly 
in the practice of credit portfolio risk management 
(see for instance \cite {Andersson-et-al, Kalkbrener-Kennedy-Popp, Puzanova-Duellmann}). 
\end{itemize}
\end{remark}

\section{Numerical analysis}\label{sec_numerical}

In this section, we numerically investigate the behavior of $\Delta \rho _\alpha ^{X, Y}$. 
Throughout this section, we assume that the distributions of $X$ and $Y$ are 
given as $\mathrm {GPD}(\xi _X, \sigma _X)$ and $\mathrm {GPD}(\xi _Y, \sigma _Y)$, respectively, 
with $\xi _X, \xi _Y \in (0, 1)$ and $\sigma _X, \sigma _Y > 0$, where 
$\mathrm {GPD}(\xi , \sigma )$ denotes the generalized Pareto distribution 
whose distribution function is given by 
$1 - ( 1 + \xi x/\sigma ) ^{-1/\xi }$, $x\geq 0$. 
Then, $\bar{F}_X$ and $\bar{F}_Y$ satisfy (\ref {regular_variation}) with $\beta = 1/\xi _X$ and $\gamma = 1/\xi _Y$. 
Note that condition [C3] is satisfied with 
\begin{eqnarray*}
k = \left( \frac{\sigma _Y}{\xi _Y}\right) ^{1/\xi _Y}\left( \frac{\sigma _X}{\xi _X}\right) ^{-1/\xi _X}
\end{eqnarray*}
(see (5.2) in \cite {Kato_IJTAF}). 
Also note that $\mathrm {VaR}_\alpha (X)$ and $\mathrm {VaR}_\alpha (Y)$ are analytically solved as 
\begin{eqnarray*}
\mathrm {VaR}_\alpha (X) = 
\frac{\sigma _X}{\xi _X}\left\{ (1-\alpha )^{-\xi _X} - 1 \right\} , \ \ 
\mathrm {VaR}_\alpha (Y) = 
\frac{\sigma _Y}{\xi _Y}\left\{ (1-\alpha )^{-\xi _Y} - 1 \right\} . 
\end{eqnarray*}
We numerically compute 
$\Delta \mathrm {VaR}^{X, Y}_\alpha , \Delta \mathrm {ES}^{X, Y}_\alpha , 
\Delta \rho ^{\mathrm {EXP}, X, Y}_\alpha , \Delta \rho ^{\mathrm {POW}, X, Y}_\alpha $, 
and $\mathrm {ES}^\mathrm {Euler}_\alpha $, where we let 
$\mathrm {ES}^\mathrm {Euler}_\alpha  = \mathrm {ES}^\mathrm {Euler}_\alpha (Y | X + Y)$ for brevity. 
In all calculations, we fix $\sigma _X = 100$ and $\sigma _Y = 80$. 
For $\xi _X$ and $\xi _Y$, we examine several patterns to study each of the following three cases: 
(i) $\beta + 1 < \gamma $, (ii) $\beta < \gamma \leq \beta + 1$, and (iii) $\beta = \gamma $. \\

\noindent 
{\it Case }(i) \ $\beta + 1 < \gamma $

We set $\xi _X = 0.5$ and $\xi _Y = 0.1$. 
Hence, $\beta = 2$ and $\gamma = 10$, so that $\beta + 1 < \gamma $ holds. 
Figure \ref {fig_1_1} shows the graphs of 
$\Delta \mathrm {ES}^{X, Y}_\alpha , 
\Delta \rho ^{\mathrm {EXP}, X, Y}_\alpha , \Delta \rho ^{\mathrm {POW}, X, Y}_\alpha $, 
and $\mathrm {ES}^\mathrm {Euler}_\alpha $. 
These values are always larger than $\E [Y]$ whenever $\alpha \in (0, 1)$, 
and they converge to $\E [Y]$ for both $\alpha \rightarrow 0$ and $\alpha \rightarrow 1$. 
Indeed, 
\begin{eqnarray}\label{limit_alpha_zero}
\lim _{\alpha \rightarrow 0}\Delta \rho ^{X, Y}_\alpha  = \E [X + Y] - \E [X] = \E [Y]
\end{eqnarray} 
holds because 
$\phi ^{\mathrm {ES}}_\alpha (u), \phi ^{\mathrm {EXP}}_\alpha (u), \phi ^{\mathrm {POW}}_\alpha (u) 
\longrightarrow 1$, $\alpha \rightarrow 0$ for each $u\in [0, 1)$. 
The limit as $\alpha \rightarrow 1$ is a consequence of Theorem \ref {th_main}. 
Moreover, the forms of these graphs are unimodal. 
That is, the function 
$\alpha \rightarrow \Delta \rho ^{X, Y}_\alpha $ 
incereases on $(0, \alpha _0)$
and decreases on $(\alpha _0, 1)$ for some $\alpha _0\in (0, 1)$. 
Intuitively, the values of $\Delta \rho ^{X, Y}_\alpha $ seem to become large 
as $\alpha $ increases because 
a larger $\alpha $ implies a greater risk sensitivity. 
However, our result implies that the impact of adding loss variable $Y$ 
into the prior risk profile $X$ is maximized at some $\alpha _0 < 1$. 

Figure \ref {fig_1_2} shows the relation between 
$\Delta \mathrm {ES}^{X, Y}_\alpha $ and $\Delta \mathrm {VaR}^{X, Y}_\alpha $. 
We see that $\Delta \mathrm {ES}^{X, Y}_\alpha $ takes a maximum at $\alpha = \alpha _0$, 
where $\alpha _0$ is a solution to 
\begin{eqnarray}\label{alpha_VaR_ES}
\Delta \mathrm {VaR}^{X, Y}_{\alpha _0} = \Delta \mathrm {ES}^{X, Y}_{\alpha _0}. 
\end{eqnarray}
Indeed, we have the following result. 

\begin{proposition}\label{prop_max_ES}
If there is a unique solution $\alpha _0\in (0, 1)$ to $(\ref {alpha_VaR_ES})$, 
then $\max _{0 < \alpha < 1}\Delta \mathrm {ES}^{X, Y}_\alpha = \Delta \mathrm {ES}^{X, Y}_{\alpha _0}$. 
\end{proposition}

Note that unlike the case of SRMs, 
$\Delta \mathrm {VaR}^{X, Y}_\alpha $ takes a value smaller than $\E [Y]$ if $\alpha $ is small. 
This is because VaR is not a convex risk measure, 
so the relation (\ref {lower_bound}) is not guaranteed for $\rho _\alpha = \mathrm {VaR}_\alpha $. 
In particular, we observe that
\begin{eqnarray}\label{limit_VaR_zero}
\lim _{\alpha \rightarrow 0}\Delta \mathrm {VaR}^{X, Y}_\alpha = 
\mathop {\rm essinf}(X+Y) - \mathop {\rm essinf}X = 0. 
\end{eqnarray}

\noindent
{\it Case }(ii) \ $\beta < \gamma \leq \beta + 1$

Figure \ref {fig_2_1} shows the approximation errors, defined as 
\begin{eqnarray}\label{def_error}
\mathrm {Error}_\alpha = \frac{\bar{M}(\alpha )}{\Delta \rho ^{X, Y}_\alpha } - 1
\end{eqnarray}
with $\xi _X = 2/3$ ($\beta = 1.5$) and $\xi _Y = 0.5$ ($\gamma = 2$). 
We see that $\mathrm {Error}_\alpha $ is close to $0$ as $\alpha \rightarrow 1$ 
for each case of $\rho _\alpha = \mathrm {ES}_\alpha , \rho ^{\mathrm {EXP}}_\alpha , \rho ^{\mathrm {POW}}_\alpha $. 
Moreover, we numerically verify the assertion of Theorem \ref {th_equivalence_Euler_EL} 
for $\rho _\alpha = \mathrm {ES}_\alpha $ in Figure \ref {fig_2_1_1}. 
We observe that $\bar{M}_\alpha / \mathrm {ES}^\mathrm {Euler}_\alpha $ converges to 
$\delta = 1/\gamma = \xi _Y = 0.5$ as $\alpha \rightarrow 1$. 

By contrast, the convergence speed of $\mathrm {Error}_\alpha $ as $\alpha \rightarrow 1$ 
decreases if the tails of $X$ and $Y$ are less fat-tailed. 
Figure \ref {fig_2_2} shows $\mathrm {Error}_\alpha $ with $\xi _X = 2/7$ ($\beta = 3.5$) and $\xi _Y = 0.25$ ($\gamma = 4$). 
We find that $\mathrm {Error}_\alpha $ decreases as $\alpha $ tends to $1$, 
but the gap between $\mathrm {Error}_\alpha $ and $0$ is still large, even in the case $\alpha = 0.999$. 

\ \\
\noindent 
{\it Case} (iii) \ $\beta = \gamma $

Finally, we look at the case $\xi _X = \xi _Y = 0.7$. 
The results are summarized in Figures \ref {fig_3_1} and \ref {fig_3_1_1}. 
We see that $\mathrm {Error}_\alpha $ approaches $0$ as $\alpha \rightarrow 1$ 
for each case of $\rho _\alpha = \mathrm {ES}_\alpha , \rho ^{\mathrm {EXP}}_\alpha , \rho ^{\mathrm {POW}}_\alpha $. 
We also confirm that $\bar{M}(\alpha )/ \mathrm {ES}^\mathrm {Euler}_\alpha $ converges to 
$\delta = \{1 + k - (1+k)^{1-1/\beta }\} / k \approx 0.870$ as $\alpha \rightarrow 1$. 

Similarly to Case (ii), 
the convergence speed of $\mathrm {Error}_\alpha $ decreases as the tails of $X$ and $Y$ 
become thinner. 
Figure \ref {fig_3_2} shows the graph of $\mathrm {Error}_\alpha $ with $\xi _X = \xi _Y = 0.3$. 
The approximation error tends to zero as $\alpha \rightarrow 1$, 
but remains smaller than $-20\% $ even when $\alpha = 0.999$.

\section{Conluding remarks}\label{sec_conclusion}

In this paper, we have studied the asymptotic behavior of the difference 
between $\rho _\alpha (X + Y)$ and $\rho _\alpha (X)$ as $\alpha \rightarrow 1$ 
when $\rho _\alpha $ is a parameterized SRM satisfying (\ref {conv_worst_case}). 
We have shown that $\Delta \rho _\alpha ^{X, Y}$ is asymptotically equivalent to $\bar{M}(\alpha )$ given by (\ref {def_M_bar}), 
whose form changes according to 
the relative magnitudes of the thicknesses of the tails of $X$ and $Y$. 
In particular, for $\beta + 1 < \gamma $, 
we found the convergence $\lim _{\alpha \rightarrow 1}\Delta \rho _\alpha ^{X, Y} = \E [Y]$ 
for general CLBSRMs $(\rho _\alpha )_\alpha $. 
Moreover, we also found that 
$\Delta \rho _\alpha ^{X, Y}\sim \delta \rho ^{\mathrm {Euler}}_\alpha (Y | X+Y)$ as $\alpha \rightarrow 1$ 
for a constant $\delta \in (0, 1]$ given by (\ref {def_delta}). 
This clarifies the asymptotic relation between the marginal risk contribution and the Euler contribution. 

Our numerical results in the case $\beta + 1 < \gamma $ showed that 
$\Delta \rho _\alpha ^{X, Y}$ is not increasing but is unimodal with respect to $\alpha $, which implies that the impact of $Y$ in the portfolio $X+Y$ does not always increase with $\alpha $. 
Interestingly, this phenomenon is inconsistent with intuition. 

Our results essentially depend on the assumption that $X$ and $Y$ are independent. 
However, 
the dependence structure of the loss variables $X$ and $Y$ plays 
an essential role in financial risk management. 
The case of dependent $X$ and $Y$ for $\rho _\alpha = \mathrm {VaR}_\alpha $ has already been studied in Section A.1 of \cite {Kato_IJTAF}. As mentioned in Remark \ref {rem_main}, we have now generalized this result to the case of CLBSRMs. 
However, we require the  somewhat strong assumption that $X$ and $Y$ are not strongly dependent on each other. 
With the additional analysis in Appendix \ref {sec_add} below, 
we will see that our main results still hold 
for a general dependence structure if $\beta + 1 < \gamma $, 
but that they are easily violated if $\beta \leq \gamma \leq  \beta + 1$. 
In future work, we will continue to study the asymptotic behavior of 
$\Delta \rho _\alpha ^{X, Y}$ as $\alpha \rightarrow 1$, 
without the independence condition.

\appendix 

\section{A short consideration of the dependent case}\label{sec_add}

Here, we briefly investigate the asymptotic behavior of $\Delta \rho _\alpha ^{X, Y}$ as $\alpha \rightarrow 1$ 
when $X$ and $Y$ are not independent. 
Throughout this section, we assume that $F_X$, $F_Y$, and $F_{X+Y}$ are continuous. 
With this, (\ref {lower_bound}) is rewritten as 
$\Delta \rho _\alpha ^{X, Y} \geq \E \hspace{0mm}^{Q^X_\alpha }[Y]$. 
Combining this result with (\ref {ineq_general}), we have 
\begin{eqnarray}\label{ineq_general2}
\E \hspace{0mm}^{Q^X_\alpha }[Y] \leq \Delta \rho _\alpha ^{X, Y} \leq \E \hspace{0mm}^{Q^{X+Y}_\alpha }[Y].
\end{eqnarray}
Note that (\ref {ineq_general2}) holds for general SRM $\rho _\alpha $ 
whenever (\ref {derivative_VaR}) holds. 

\subsection{Comonotonic case}\label{sec_comonotone}

We consider the case where $X$ and $Y$ are comonotone. 
In other words, they are perfectly positively dependent (see Definition 4.82 of \cite {Foellmer-Schied} and Definition 5.15 in \cite {McNeil-Frey-Embrechts}). 
In this case, the following proposition is straightforwardly shown. 

\begin{proposition}\label{prop_comonotone}
If $X$ and $Y$ are comonotone, then 
\begin{eqnarray}\label{rel_comonotone}
\Delta \rho _\alpha ^{X, Y} = \E \hspace{0mm}^{Q^X_\alpha }[Y] = 
\E \hspace{0mm}^{Q^{X+Y}_\alpha }[Y] = \rho _\alpha (Y). 
\end{eqnarray}
\end{proposition}

This proposition implies that when $\beta + 1 < \gamma $, 
the asymptotic relations (\ref {case_weak_dependent}) and (\ref {asy_eq_contribution}) still hold, 
even if $X$ and $Y$ are strongly correlated, 
but that the assertions of Theorems \ref {th_main} and \ref {th_equivalence_Euler_EL} 
do not necessarily hold when $\beta \leq \gamma \leq \beta + 1$.

\subsection{Additional numerical analysis}
Similarly to Section \ref {sec_numerical}, 
we assume that $X\sim \mathrm {GPD}(\xi _X, \sigma _X)$ and $Y\sim \mathrm {GPD}(\xi _Y, \sigma _Y)$ with 
$\sigma _X = 100$, $\sigma _Y = 80$. 
To describe the dependence between $X$ and $Y$, we introduce a copula. 
By Sklar's theorem, we see that the joint distribution function $F_{(X, Y)}(x, y) = P(X \leq x, Y\leq y)$ of the random vector $(X, Y)$ 
is represented by 
\begin{eqnarray*}
F_{(X, Y)}(x, y) = C(F_X(x), F_Y(y)), 
\end{eqnarray*}
for a copula $C: [0, 1]^2\longrightarrow [0, 1]$, 
which is a distribution function with uniform marginals. 
Here, we examine the following three copulas: 
\begin{itemize}
 \item [(a)] The Gaussian copula $C^\mathrm {Gauss}_\rho (u, v) = \Phi (\Phi ^{-1}(u), \Phi ^{-1}(v))$, $-1 < \rho < 1$, 
 \item [(b)] The Gumbel copula $C^\mathrm {Gumbel}_\theta (u, v) = \exp \left( -((-\log u)^\theta + (-\log v)^\theta )^{1/\theta }\right) $, $\theta \geq 1$, 
 \item [(c)] The countermonotonic copula $C^\mathrm {cmon}(u, v) = \max \{u + v - 1, 0\}$, 
\end{itemize}
where $\Phi (x) = \int ^x_{-\infty }e^{-y^2/2}/\sqrt{2\pi }dy$ is the distribution function of the standard normal distribution 
(for more details on the copulas, see, for instance, Chapter 5 of \cite {McNeil-Frey-Embrechts}). 
The parameters $\rho $ in (a) and $\theta $ in (b) describe the strength of the dependence between $X$ and $Y$. 
We always set $\rho = 0.3$ and $\theta = 3$ in this section. 
If $C = C^\mathrm {cmon}$, then $X$ and $Y$ are perfectly negatively dependent. 
In particular, in that case, $X$ and $Y$ are represented as $X = F^{-1}_X(U)$ and $Y = F^{-1}_Y(1 - U)$, 
where $U$ is a random variable with uniform distribution on $(0, 1)$. 

Figure \ref {fig_A_1} summarizes the results with $\xi _X = 0.5$ and $\xi _Y = 0.1$. 
We compare the values of $\Delta \rho ^{X, Y}_\alpha $ 
(with $\rho _\alpha = \mathrm {ES}_\alpha , \rho ^\mathrm {EXP}_\alpha , \rho ^\mathrm {POW}_\alpha $) 
and $\mathrm {ES}^\mathrm {Euler}_\alpha (Y | X + Y)$. 
We find that all these values converge to the same value, 
which is not equal to $\E [Y]$, 
by letting $\alpha \rightarrow 1$. 
Note that when $X$ and $Y$ are countermonotonic, 
they converge to zero as $\alpha \rightarrow 1$, 
so (\ref {cond_liminf}) does not hold in this case. 

Figure \ref {fig_A_2} shows the graphs of the relative errors defined by (\ref {def_error}) 
with $\rho _\alpha = \mathrm {ES}_\alpha , \rho ^\mathrm {EXP}_\alpha , \rho ^\mathrm {POW}_\alpha $ 
when we set $\xi _X = 2/3$ and $\xi _Y = 0.5$. 
We find that $\mathrm {Error}_\alpha $ does not converge to zero as $\alpha \rightarrow 1$. 
Similar phenomena are observed in Figure \ref {fig_A_3} with the settings $\xi _X = \xi _Y = 0.7$. 
Therefore, the assertion of Theorem \ref {th_main} does not hold when $\beta \leq \gamma \leq \beta + 1$ if $X$ and $Y$ are correlated. 

Note that the above findings are consistent with the comonotonic case (Proposition \ref {prop_comonotone}).

\subsection{Theoretical result in the case $\beta + 1 < \gamma $}

We describe the following conditions. 
\begin{itemize}
 \item[ \mbox{[C5]}] \ For each $y\geq 0$, 
$F_X(\cdot | Y = y)$ has a positive, non-increasing density function 
$f_X(\cdot | Y = y)$ on $[0, \infty )$, where 
$F_X(\cdot | Y = y)$ is the conditional distribution function of $X$ given $Y = y$. 
Moreover, $f_X(x | Y = y)$ is continuous in $x$ and $y$. 
 \item[ \mbox{[C6]}] \ There is a $\kappa \in \mathbb {R}$ such that 
$f_X(x | Y = y)$ is uniformly regularly varying with index $\kappa $ in the following sense: 
\begin{eqnarray}\label{cond_D}
\mathop {\sup}_{y\geq 0}\left| 
\frac{f_X(tx | Y = y)}{f_X(x | Y = y)} - t^\kappa \right| \ \longrightarrow \ 0, \ \ 
x\rightarrow \infty 
\end{eqnarray}
for each $t > 0$. 
Moreover, $f_{X+Y}$ is regularly varying. 
 \item[ \mbox{[C7]}] \ It holds that 
\begin{eqnarray}\label{cond_D2}
\sup _{x\geq 0}\E [Y^\eta | X = x] + 
\sup _{z\geq 0}\E [Y^\eta | X + Y = z] < \infty 
\end{eqnarray}
for some $\eta > \max \{ -\kappa - \beta , 1 \}$. 
\end{itemize}

Conditions [C5]--[C7] strongly correspond to conditions [A5]--[A6] in \cite {Kato_IJTAF}. 
It should be noted that the index parameter $\kappa $ is assumed to be equal to $-\beta - 1$ in condition [A6] in  \cite {Kato_IJTAF}, 
but that this equality is not required to obtain our results. 
Note also that 
$\kappa $ may be different from $-\beta - 1$. 
Indeed, we can verify, at least numerically, that for each $y\geq 0$, 
the function $f_X(\cdot | Y = y)$ is regularly varying with 
index $\kappa = - 1 - \beta /(1 - \rho ^2)$ 
(resp., $\kappa = -\theta \beta - 1$) if we adopt 
$C = C^\mathrm {Gauss}_\rho $ (resp., $C = C^\mathrm {Gumbel}_\theta $) as a copula for the random vector $(X, Y)$
whose marginal distributions are given by the generalized Pareto distribution. 

Using a similar argument as in the proof of the uniform convergence theorem (Theorem 1.2.1 in \cite {Bingham-Goldie-Teugels}), together with the continuity of $f_X(x | Y = y)$ in $y$,
we get from (\ref {cond_D}) that 
\begin{eqnarray}\label{cond_D_unif}
\sup _{t\in K}\mathop {\sup}_{y\geq 0}\left| 
\frac{f_X(tx | Y = y)}{f_X(x | Y = y)} - t^\kappa \right| \ \longrightarrow \ 0, \ \ 
x\rightarrow \infty,
\end{eqnarray}
for each compact set $K\subset (0, \infty )$. 

We now introduce the following result. 
\begin{theorem}\label{th_dependent}Assume {\rm [C5]--[C7]} and {\rm (\ref {cond_liminf})}. 
If $\beta + 1 < \gamma $, it holds that 
\begin{eqnarray*}
\E \hspace{0mm}^{Q^X_\alpha }[Y] \sim \Delta \rho ^{X, Y}_\alpha \sim 
\E \hspace{0mm}^{Q^{X+Y}_\alpha }[Y], \ \ \alpha \rightarrow 1. 
\end{eqnarray*}
\end{theorem}

This theorem claims that both (\ref {case_weak_dependent}) and (\ref {asy_eq_contribution}) are true 
under some conditions, even when $X$ and $Y$ are dependent.

\section{Proofs}\label{sec_proofs}

\begin{proof}[Proof of Lemma \ref {lem_pointwise}]
Assume (\ref {def_CLBSRM}). 
Fix any $u\in [0, 1)$. 
Then, (\ref {def_CLBSRM}) implies that 
\begin{eqnarray}\label{temp_lemma1}
\lim _{\alpha \rightarrow 1}\int ^{(1 + u)/2}_u\phi _\alpha (v)dv = 0. 
\end{eqnarray}
Because $\phi _\alpha $ is non-decreasing and non-negative, we see that 
\begin{eqnarray}\label{temp_lemma2}
\int ^{(1 + u)/2}_u\phi _\alpha (v)dv \geq \frac{1 - u}{2}\phi _\alpha (u) \geq 0. 
\end{eqnarray}
Combining (\ref {temp_lemma1}) with (\ref {temp_lemma2}), we have 
$\lim _{\alpha \rightarrow 0}\phi _\alpha (u) = 0$. 

Conversely, if we assume (\ref {conv_pointwise}), 
then Prokhorov's theorem implies that 
for each increasing sequence $(\alpha _n)_{n\geq 1}\subset (0, 1)$ with $\lim _n\alpha _n = 1$ 
there is a further subsequence $(\alpha _{n_k})_{k\geq 1}$ and a probability measure $\mu $ 
on $[0, 1]$ such that 
$\Phi _{\alpha _{n_k}}$ weakly converges to $\mu $ as $k\rightarrow \infty $. 
Then, for each $\beta \in (0, 1)$, we see that 
\begin{eqnarray*}
0\leq \mu ([0, \beta )) \leq  \liminf _{k\rightarrow \infty }\int ^\beta _0\phi _{\alpha _{n_k}}(u)du 
\leq 
\liminf _{k\rightarrow \infty }\beta \phi _{\alpha _{n_k}}(\beta ) = 0. 
\end{eqnarray*}
This immediately leads us to $\mu ([0, 1)) = 0$, hence $\mu = \delta _1$. 
We therefore arrive at (\ref {def_CLBSRM}). 
\end{proof}

\begin{proof}[Proof of Proposition \ref {prop_max_ES}]
Let $f(\alpha ) = \Delta \mathrm {ES}^{X, Y}_\alpha $. 
We observe that 
\begin{eqnarray*}
f'(\alpha )= \frac{1}{(1-\alpha )^2}\int ^1_\alpha \Delta \mathrm {VaR}^{X, Y}_udu - 
\frac{1}{1 - \alpha }\Delta \mathrm {VaR}^{X, Y}_\alpha  = 
\frac{g(\alpha )}{1 - \alpha }, 
\end{eqnarray*}
where $g(\alpha ) = \Delta \mathrm {ES}^{X, Y}_\alpha - \Delta \mathrm {VaR}^{X, Y}_\alpha $. 
By (\ref {limit_alpha_zero}), (\ref {limit_VaR_zero}), and Theorem \ref {th_main}, 
we see that $g$ is continuous on $(0, 1)$, $g(0+) = \E [Y] > 0$ and $g(1-) = 0$. 
Moreover, by the assumption, it holds that $g(\alpha _0) = 0$ and 
$g(\alpha ) \neq 0$ for all $\alpha \in (0, 1)\setminus \{\alpha _0\}$. 
Together, these imply that $g$ is positive on $(0, \alpha _0)$ and negative on $(\alpha _0, 1)$, and that $f'$ has the same pattern. 
Therefore, $f(\alpha )$ takes a maximum at $\alpha = \alpha _0$. 
\end{proof}

\begin{proof}[Proof of Proposition \ref {prop_comonotone}]
Because $\rho _\alpha $ is comonotonic, 
we obviously have 
\begin{eqnarray*}
\Delta \rho ^{X, Y}_\alpha = \rho _\alpha (X) + \rho _\alpha (Y) - \rho _\alpha (X) = \rho _\alpha (Y). 
\end{eqnarray*}
Here, we see that $X = F^{-1}_X(U)$ and $Y = F^{-1}_Y(U)$ for some random variable $U$ 
with uniform distribution on $(0, 1)$ 
(see Lemmas 4.89--4.90 in \cite {Foellmer-Schied} and their proofs). 
Then we have 
\begin{eqnarray*}
\E [Y | X = \mathrm {VaR}_\alpha (X)] = F^{-1}_Y(\alpha ) = \mathrm {VaR}_\alpha (Y),
\end{eqnarray*}
and thus 
\begin{eqnarray*}
\E \hspace{0mm}^{Q^X_\alpha }[Y] = \int ^1_0\mathrm {VaR}_u(Y)\phi _\alpha (u)du = \rho _\alpha (Y). 
\end{eqnarray*}
Similarly, because $F^{-1}_{X+Y} = F^{-1}_X + F^{-1}_Y$, we have 
\begin{eqnarray*}
\E [Y | X + Y = \mathrm {VaR}_\alpha (X + Y)] = \mathrm {VaR}_\alpha (Y),
\end{eqnarray*}
and 
$\E \hspace{0mm}^{Q^{X+Y}_\alpha }[Y] = \rho _\alpha (Y)$, which completes the proof. 
\end{proof}

\subsection{Proof of Theorem \ref {th_main}}

We first state some propositions and prove them. 
For this, let $\bar{f}$ be given as (\ref {def_f_bar}). 
Note again that $\bar{M}$ defined in (\ref {def_M_bar}) satisfies 
\begin{eqnarray*}
\bar{M}(\alpha ) = \int ^1_0\bar{f}(u)\phi _\alpha (u)du. 
\end{eqnarray*}

\begin{proposition} \label{prop_lower_bdd} 
$\liminf _{\alpha \rightarrow 1}\bar{M}(\alpha ) > 0$. 
\end{proposition}

\begin{proof}
If $\beta + 1 < \gamma $, we see that 
$\bar{f}(\alpha ) = \E [Y] > 0$ 
because $Y$ is non-negative and $\bar{F}_Y$ is positive. 
If $\beta < \gamma \leq \beta + 1$, we observe 
\begin{eqnarray*}
\bar{M}(\alpha )
&\geq & 
\frac{k}{\beta }\int ^1_{\alpha _0}\mathrm {VaR}_u(X)^{\beta + 1 - \gamma }\phi _\alpha (u)du,\\
&\geq & 
\frac{k}{\beta }\mathrm {VaR}_{\alpha _0}(X)^{\beta + 1 - \gamma }\left( 1 - \int ^{\alpha _0}_0\phi _\alpha (u)du\right), \\
&\longrightarrow & 
\frac{k}{\beta }\mathrm {VaR}_{\alpha _0}(X)^{\beta + 1 - \gamma } > 0, \ \ \alpha \rightarrow 1, 
\end{eqnarray*}
where $\alpha _0\in (0, 1)$ is a real number satisfying $\mathrm {VaR}_{\alpha _0}(X) > 0$. The existence of such an $\alpha _0$ can be proven using Propositions 1.5.1 and 1.5.15 in 
\cite {Bingham-Goldie-Teugels}. 
Similarly, if $\beta = \gamma $, we have 
\begin{eqnarray*}
\liminf _{\alpha \rightarrow 1}\bar{M}(\alpha )\geq 
\{(1 + k)^{1/\beta } - 1\}\mathrm {VaR}_{\alpha _0}(X) > 0.  \qedhere 
\end{eqnarray*}
\end{proof}

\begin{proposition} \label{prop_int_f_bar} 
$0 \leq \int ^1_0\bar{f}(u)du < \infty $. 
\end{proposition}

\begin{proof} 
If $\beta + 1 < \gamma $, 
the assertion is obvious from the assumption $\beta > 1$. 
If $\beta < \gamma \leq \beta + 1$, we see that 
\begin{eqnarray*}
0\leq \int ^1_0\bar{f}(u)du = \frac{k}{\beta }\E [X^{\beta + 1 - \gamma }] \leq \frac{k}{\beta }\E [X]^{\beta + 1 - \gamma } < \infty, 
\end{eqnarray*}
because of $0 < \beta + 1 - \gamma < 1$. 
If $\beta = \gamma $, we have 
\begin{eqnarray*}
0\leq \int ^1_0\bar{f}(u)du = 
\{(1 + k)^{1/\beta } - 1\}\E [X] < \infty .  \qedhere 
\end{eqnarray*}
\end{proof}

\begin{cor}\label{cor_finiteness_M}$\bar{M}(\alpha ) < \infty $, $\alpha \in (0, 1)$. 
\end{cor}

\begin{proof}
This follows from (\ref {ass_bdd_phi}) and Proposition \ref {prop_int_f_bar}. 
\end{proof}

\noindent 
\begin{proof}[\it Proof of Theorem \ref {th_main}]
Let 
$f(\alpha ) = \Delta \mathrm {VaR}^{X,Y}_\alpha $, $\alpha \in (0, 1)$. 
Note that 
\begin{eqnarray}\label{conv_EL}
\lim _{\alpha \rightarrow 1}\frac{f(\alpha )}{\bar{f}(\alpha )} = 1,
\end{eqnarray}
by virtue of Theorem 4.1(i)--(iii) in \cite {Kato_IJTAF}. 
Moreover, (\ref {conv_EL}) immediately implies 
\begin{eqnarray}\label{unif_conv_EL}
\lim _{\alpha \rightarrow 1}\sup _{u\in [\alpha , 1)}\left| \frac{f(u)}{\bar{f}(u)} - 1\right| = 0. 
\end{eqnarray}
Furthermore, it holds that 
\begin{eqnarray}\label{f_integrability}
\int ^1_0f(u)du = \int ^1_0\mathrm {VaR}_u(X+Y)du - \int ^1_0\mathrm {VaR}_u(X)du = 
\E [X + Y] - \E [X] = \E [Y] < \infty ,  \ \ \ 
\end{eqnarray}
hence $f$ is integrable. 
The integrability of $\bar{f}$ is guaranteed by Proposition \ref {prop_int_f_bar}. 

Temporarily fix any $\delta \in (0, 1)$. 
From (\ref {conv_pointwise}) and (\ref {f_integrability}), we easily see that 
\begin{eqnarray}\label{conv_f}
(0\leq ) \int ^{\delta }_0f(u)\phi _\alpha (u)du 
\leq \phi _\alpha (\delta )\int ^1_0f(u)du \longrightarrow 0, \ \ \alpha \rightarrow 1. 
\end{eqnarray}
Similarly, we have 
\begin{eqnarray}\label{conv_bar_f}
\lim _{\alpha \rightarrow 1}\int ^{\delta }_0\bar{f}(u)\phi _\alpha (u)du = 0. 
\end{eqnarray}

Additionally, we have 
\begin{eqnarray}\label{temp_th1_target}
\frac{\int ^1_{\delta }f(u)\phi _\alpha (u)du}{\bar{M}(\alpha )} = 
\int ^1_{\delta }\frac{f(u) }{\bar{f}(u)}\psi _\alpha (u)du, 
\end{eqnarray}
where 
$\psi _\alpha (u) = \bar{f}(u)\phi _\alpha (u)/\bar{M}(\alpha )$. 
Using (\ref {conv_bar_f}) and Proposition \ref {prop_lower_bdd}, we obtain 
\begin{eqnarray}\label{calc_A}
&&\int ^1_{\delta }\psi _\alpha (u)du = 1 - \frac{1}{\bar{M}(\alpha )}\int ^{\delta }_0\bar{f}(u)\phi _\alpha (u)du \longrightarrow 1, \ \ \alpha  \rightarrow 1. 
\end{eqnarray}
By 
(\ref {temp_th1_target}) and (\ref {calc_A}), we have 
\begin{eqnarray*}
&&\left| \frac{\int ^1_{\delta }f(u)\phi _\alpha (u)du}{\bar{M}(\alpha )} - 1 \right| \leq 
\left| 
\int ^1_{\delta }
\left( \frac{f(u)}{\bar{f}(u)} - 1\right) \psi _\alpha (u)du \right|  + 
\left| \int ^1_{\delta }\psi _\alpha (u)du - 1 \right| \\
&\leq & 
\sup _{u\in [\delta , 1)}\left| \frac{f(u)}{\bar{f}(u)} - 1\right|  + 
\left| \int ^1_{\delta }\psi _\alpha (u)du - 1 \right| 
\ \longrightarrow \ 
\sup _{u\in [\delta , 1)}\left| \frac{f(u)}{\bar{f}(u)} - 1\right| , \ \ \alpha \rightarrow 1. 
\end{eqnarray*}
Combining this with (\ref {conv_f}) and Proposition \ref {prop_lower_bdd}, we arrive at 
\begin{eqnarray*}
\limsup _{\alpha \rightarrow 1}\left| \frac{\Delta \rho ^{X, Y}_\alpha }{\bar{M}(\alpha )} - 1\right| 
\leq \sup _{u\in [\delta , 1)}\left| \frac{f(u)}{\bar{f}(u)} - 1\right| . 
\end{eqnarray*}
Because $\delta \in (0, 1)$ is arbitrary, we obtain the desired assertion by (\ref {unif_conv_EL}). 
\end{proof}

\subsection{Proof of Theorem \ref {th_equivalence_Euler_EL}}

Let $Z = X + Y$ for brevity. 
We see that $Z$ has a density function 
\begin{eqnarray*}
f_Z(z) = \int ^z_0f_X(z - y)f_Y(y)dy = \int ^z_0f_X(x)f_Y(z - x)dx. 
\end{eqnarray*}

\begin{lemma}\label{lem_fZ}
$f_Z$ is positive and continuous on $(0, \infty )$. 
Moreover, $f_Z$ is regularly varying with index $-\min \{ \beta , \gamma  \} - 1$ and it holds that 
\begin{eqnarray}\label{mdt_Z}
\lim _{z\rightarrow \infty }\frac{zf_Z(z)}{\bar{F}_Z(z)} = \min \{ \beta , \gamma  \} . 
\end{eqnarray}
\end{lemma}

\begin{proof}
Continuity and positivity are obvious. 
By [C4] and Theorem 1.1 in \cite {Bingham-Goldie-Omey}, we see that 
$f_Z(z) \sim f_X(z) + f_Y(z)$, $z\rightarrow \infty $ and that 
$f_Z$ is regularly varying with index 
$\max \{-\beta - 1, -\gamma - 1\} = -\min \{ \beta , \gamma  \} - 1$. 
The last assertion is obtained by Proposition 1.5.10 in \cite {Bingham-Goldie-Teugels}. 
\end{proof}

Let $F_Y(\cdot | Z = z)$ be the conditional distribution function of $Y$ given $Z = z$. 
Then we have 
\begin{eqnarray}\label{cond_exp1}
\E [Y | Z = z] = \int ^\infty _0yF_Y(dy | Z = z). 
\end{eqnarray}

\begin{proposition}\label{prop_cond_dist_Y}It holds that 
\begin{eqnarray*}
F_Y(y | Z = z) = \int ^{y\wedge z}_0\frac{f_X(z - y')}{f_Z(z)}f_Y(y')dy', \ \ y, z\geq 0. 
\end{eqnarray*}
\end{proposition}

\begin{proof}
For each $y, z\geq 0$, a straightforward calculation gives 
\begin{eqnarray*}
\int ^z_0\int ^{y\wedge z'}_0\frac{f_X(z' - y')}{f_Z(z')}f_Y(y')dy'f_Z(z')dz' = 
P(Y\leq y, \ Z\leq z), 
\end{eqnarray*}
which implies our assertion. 
\end{proof}

Note that (\ref {cond_exp1}) and Proposition \ref {prop_cond_dist_Y} lead to 
\begin{eqnarray}\label{cond_exp2}
\E [Y | Z = z] = \int ^z_0y\frac{f_X(z - y)}{f_Z(z)}f_Y(y)dy. 
\end{eqnarray}

\begin{proposition}\label{prop_equivalence_Euler_EL}
If $\beta + 1 < \gamma $, then 
\begin{eqnarray*}
\E [Y | Z = \mathrm {VaR}_\alpha (Z)]\longrightarrow \E [Y], \ \ \alpha \rightarrow 1. 
\end{eqnarray*}
\end{proposition}

\begin{proof}
Let 
\begin{eqnarray}\label{def_z}
z_\alpha &=& \mathrm {VaR}_\alpha (Z), \\\label{def_Gy}
G_\alpha (y) &=& y\frac{f_X(z_\alpha - y)}{f_Z(z_\alpha )}1_{[0, z_\alpha /2]}(y), \\\label{def_Hx}
H_\alpha (x) &=& (z_\alpha - x)\frac{f_Y(z_\alpha - x)}{f_Z(z_\alpha )}1_{[0, z_\alpha /2)}(x). 
\end{eqnarray}
Then, we see that 
\begin{eqnarray}\label{rep_GH}
\E [G_\alpha (Y)] + \E [H_\alpha (X)] = 
\left( \int ^{z_\alpha /2}_0 + \int ^{z_\alpha }_{z_\alpha /2}\right) 
y\frac{f_X(z_\alpha  - y)}{f_Z(z_\alpha )}f_Y(y)dy = \E [Y | Z = z_\alpha ]. 
\end{eqnarray}
Therefore, we need to show that 
\begin{eqnarray}\label{conv_G_alpha}
\E [G_\alpha (Y)] \longrightarrow \E [Y], \ \ 
\E [H_\alpha (X)] \longrightarrow 0,\ \ \alpha \rightarrow 1. 
\end{eqnarray}

First, we show that 
\begin{eqnarray}\label{conv_G}
G_\alpha (y) \longrightarrow y, \ \ 
H_\alpha (x) \longrightarrow 0, \ \ \alpha \rightarrow 1 \ \ \mbox { for each } \ x, y\geq 0. 
\end{eqnarray}
Using (\ref {mdt_Z}), Lemmas A.1 and A.3 in \cite {Kato_IJTAF}, and 
Proposition A3.8 in \cite {Embrechts-Klueppelberg-Mikosch}, we obtain 
\begin{eqnarray*}
\frac{f_X(z_\alpha - y)}{f_Z(z_\alpha )} = 
\frac{f_X(z_\alpha - y)}{f_X(z_\alpha )}\cdot 
\frac{z_\alpha f_X(z_\alpha )}{\bar{F}_X(z_\alpha )}\cdot 
\frac{\bar{F}_X(z_\alpha )}{\bar{F}_Z(z_\alpha )}\cdot 
\frac{\bar{F}_Z(z_\alpha )}{z_\alpha f_Z(z_\alpha )} \longrightarrow 
1\cdot \beta \cdot 1\cdot \frac{1}{\beta } = 1, \ \ \alpha \rightarrow 1.
\end{eqnarray*}
Furthermore, 
we observe that
\begin{eqnarray}\label{bdd_H}
0\leq H_\alpha (x) \leq \frac{z_\alpha f_Y(z_\alpha / 2)}{f_Z(z_\alpha )}, 
\end{eqnarray}
and that the function $z\mapsto zf_Y(z/2)/f_Z(z)$ is regulary varying with index $\beta + 1 - \gamma < 0$. 
Thus, we obtain 
\begin{eqnarray*}
\frac{z_\alpha f_Y(z_\alpha / 2)}{f_Z(z_\alpha )} \longrightarrow 0, \ \ \alpha \rightarrow 1. 
\end{eqnarray*}
Now, (\ref {conv_G}) is obvious. 

Next, we observe that
\begin{eqnarray*}
0\leq G_\alpha (Y) + H_\alpha (X) \leq Y\frac{f_X(z_\alpha / 2)}{f_Z(z_\alpha )} + 
\frac{z_\alpha f_Y(z_\alpha / 2)}{f_Z(z_\alpha )}. 
\end{eqnarray*}
Because $(f_X(z_\alpha /2)/f_Z(z_\alpha ))_{\alpha }$ and 
$(z_\alpha f_Y(z_\alpha /2)/f_Z(z_\alpha ))_{\alpha }$ are convergent (as $\alpha \rightarrow 1$), 
they are bounded. 
Thus, we have 
\begin{eqnarray}\label{temp_bdd}
0\leq G_\alpha (Y) + H_\alpha (X) \leq C(Y + 1)
\end{eqnarray}
for some $C > 0$. 
By (\ref {conv_G}) and (\ref {temp_bdd}), 
we can apply the dominated convergence theorem to obtain (\ref {conv_G_alpha}). 
\end{proof}

\begin{proposition}\label{prop_equivalence_Euler_EL_sp}
If $\beta + 1 = \gamma $, then 
\begin{eqnarray*}
\E [Y | Z = \mathrm {VaR}_\alpha (Z)]\longrightarrow \E [Y] + \frac{k\gamma }{\beta }, \ \ \alpha \rightarrow 1. 
\end{eqnarray*}
\end{proposition}

\begin{proof}
Let $z_\alpha $, $G_\alpha (y)$, and $H_\alpha (x)$ be the same as in 
(\ref {def_z})--(\ref {def_Hx}). 
First, we have $\E [G_\alpha (Y)]\longrightarrow \E [Y]$, $\alpha \rightarrow 1$ 
by the same argument as in the proof of Proposition \ref {prop_equivalence_Euler_EL}. 
Next, for each $x\geq 0$, we see that 
\begin{eqnarray*}
H_\alpha (x) &=& 
\frac{(z_\alpha - x)f_Y(z_\alpha - x)}{\bar{F}_Y(z_\alpha - x)}\cdot 
\frac{\bar{F}_Y(z_\alpha - x)}{\bar{F}_Y(z_\alpha )}\cdot 
\frac{z_\alpha \bar{F}_Y(z_\alpha )}{\bar{F}_X(z_\alpha )}\\
&&\times 
\frac{\bar{F}_X(z_\alpha )}{\bar{F}_Z(z_\alpha )}\cdot 
\frac{\bar{F}_Z(z_\alpha )}{z_\alpha f_Z(z_\alpha )}1_{[0, z_\alpha /2)}(x)\\
&\longrightarrow & 
\gamma \cdot 1\cdot k\cdot 1\cdot \frac{1}{\beta }\cdot 1 = \frac{k\gamma }{\beta }, \ \ \alpha \rightarrow 1 
\end{eqnarray*}
due to [C3], (\ref {mdt_Z}), Proposition A3.8 in \cite {Embrechts-Klueppelberg-Mikosch}, 
Proposition 3.1(i) in \cite {Kato_IJTAF}, and Lemmas A.1 and A.3 in \cite {Kato_IJTAF}. 
Moreover, we have (\ref {bdd_H}), and the right-hand side of this inequality converges to 
$2^{\gamma + 1}k\gamma / \beta $ as $\alpha \rightarrow 1$, and so it is bounded. 
Therefore, we apply the dominated convergence theorem to obtain 
$\E [H_\alpha (X)]\longrightarrow k\gamma /\beta $ as $\alpha \rightarrow 1$. 
We complete the proof by combining these with (\ref {rep_GH}). 
\end{proof}

\begin{proposition}\label{prop_equivalence_Euler_EL2}
If $\beta < \gamma < \beta + 1$, we have 
\begin{eqnarray*}
\E [Y | Z = \mathrm {VaR}_\alpha (Z)]\sim 
\frac{k\beta }{\gamma }\mathrm {VaR}_\alpha (X)^{\beta + 1 - \gamma }, \ \ \alpha \rightarrow 1. 
\end{eqnarray*}
\end{proposition}

\begin{proof}
Let $z_\alpha $, $G_\alpha (y)$ and $H_\alpha (x)$ be set as earlier. 
Similarly to the proof of Propositions \ref {prop_equivalence_Euler_EL} and \ref {prop_equivalence_Euler_EL_sp}, 
we get $\E [G_\alpha (Y)]\longrightarrow \E [Y]$, $\alpha \rightarrow 1$. 
This implies that $\E [G_\alpha (Y)]/x^{\beta + 1 - \gamma }_\alpha \longrightarrow 0$, $\alpha \rightarrow 1$, 
where $x_\alpha = \mathrm {VaR}_\alpha (X)$. 
Therefore, it suffices to show that
$\E [H_\alpha (X)]/x^{\beta + 1 - \gamma }_\alpha \longrightarrow k\beta /\gamma $ as $\alpha \rightarrow 1$, 
which is easy to see by using similar calculations as in the proof of Proposition \ref {prop_equivalence_Euler_EL_sp} 
and by using Proposition 3.1(i) in \cite {Kato_IJTAF}. 
\end{proof}

\begin{proposition}\label{prop_equivalence_Euler_EL3}
If $\beta = \gamma $, then 
\begin{eqnarray*}
\E [Y | Z = \mathrm {VaR}_\alpha (Z)]\sim k(1+k)^{-1+1/\beta }
\mathrm {VaR}_\alpha (X), \ \ \alpha \rightarrow 1. 
\end{eqnarray*}
\end{proposition}

\begin{proof}
Similarly to the proof of Proposition \ref {prop_equivalence_Euler_EL2}, 
we need to show only that
\begin{eqnarray}\label{temp_target2}
\E [\tilde{H}_\alpha (X)] \longrightarrow k(1+k)^{-1+1/\beta }, \ \ \alpha \rightarrow 1, 
\end{eqnarray}
where $\tilde{H}_\alpha (x) = H_\alpha (x) / x_\alpha $. 
Note that 
Lemmas A.1 and A.2 in \cite {Kato_IJTAF} imply 
$\bar{F}_Z(x)\sim \bar{F}_X(x) + \bar{F}_Y(x)\sim (1+k)\bar{F}_X(x)\sim (k^{-1}+1)\bar{F}_Y(x)$, $x \rightarrow \infty $ 
and $z_\alpha \sim (1+k)^{1/\beta }x_\alpha $, $\alpha \rightarrow 1$. 
Therefore, for each $x\geq 0$, we observe 
\begin{eqnarray*}
\tilde{H}_\alpha (x) &=& 
\frac{z_\alpha }{x_\alpha }\cdot 
\frac{(z_\alpha - x)f_Y(z_\alpha - x)}{\bar{F}_Y(z_\alpha - x)}\cdot 
\frac{\bar{F}_Y(z_\alpha - x)}{\bar{F}_Y(z_\alpha )}\cdot 
\frac{\bar{F}_Y(z_\alpha )}{\bar{F}_Z(z_\alpha )}\cdot 
\frac{\bar{F}_Z(z_\alpha )}{z_\alpha f_Z(z_\alpha )}1_{[0, z_\alpha /2)}(x)\\
&\longrightarrow & 
(1+k)^{1/\beta }\cdot \beta \cdot 1\cdot \frac{k}{1+k}\cdot \frac{1}{\beta }\cdot 1 = k(1+k)^{-1+1/\beta }, \ \ \alpha \rightarrow 1 
\end{eqnarray*}
by [C3], (\ref {mdt_Z}), Proposition A3.8 in \cite {Embrechts-Klueppelberg-Mikosch}, and Lemma A.3 in \cite {Kato_IJTAF}. 
Moreover, we have 
\begin{eqnarray*}
0\leq \tilde{H}_\alpha (x) \leq \frac{z_\alpha }{x_\alpha }\frac{f_Y(z_\alpha /2)}{f_Z(z_\alpha )} 
\longrightarrow 2^{\gamma + 1}k(1+k)^{-1+1/\beta }, \ \ \alpha \rightarrow 1, 
\end{eqnarray*}
and thus we obtain (\ref {temp_target2}) by applying the dominated convergence theorem. 
\end{proof}

\begin{proof}[Proof of Theorem \ref {th_equivalence_Euler_EL}] 
We can verify that the random vector $(X, Y)$ satisfies (a version of) Assumption (S) in \cite {Tasche2000} 
by using a standard argument. 
Therefore, (\ref {derivative_VaR}) is true from (5.13) in \cite {Tasche2000}. 
Additionally, using Propositions \ref {prop_equivalence_Euler_EL}--\ref {prop_equivalence_Euler_EL3}, 
we see that for each $\varepsilon > 0$, 
there is an $\alpha _0 \in (0, 1)$ such that 
\begin{eqnarray}\label{est_cond_VaR}
\left |\frac{\delta \bar{g}(\alpha )}{\bar{f}(\alpha )} - 1\right | < \varepsilon , \ \ \alpha \in [\alpha _0, 1), 
\end{eqnarray}
where we denote $\bar{g}(\alpha ) = \E [Y | Z = \mathrm {VaR}_\alpha (Z)]$. 
Moreover, it is easy to see that $\bar{f}$ and $\bar{g}$ are bounded on $[0, \alpha _0]$. 
Therefore, combining (\ref {derivative_VaR}), (\ref {Euler_rho}), and (\ref {est_cond_VaR}), we get 
\begin{eqnarray*}
&&\left| \frac{\delta \rho ^\mathrm {Euler}_\alpha (Y | Z)}{\bar{M}(\alpha )} - 1 \right| \\
&\leq & 
\frac{1}{\bar{M}(\alpha )}
\left\{ \left(\delta \sup _{u\in [0, \alpha _0]}\bar{g}(u) + \sup _{u\in [0, \alpha _0]}\bar{f}(u)\right) \int ^{\alpha _0}_0 \phi _\alpha (u)du + 
\varepsilon \int ^1_{\alpha _0}\bar{f}(u)\phi _\alpha (u)du\right\} \\
&\longrightarrow & \frac{\varepsilon }{\delta _0}, \ \ \alpha \rightarrow 1, 
\end{eqnarray*}
where $\delta _0 = \liminf _{\alpha \rightarrow 1}\bar{M}(\alpha )$, which is positive due to Proposition \ref {prop_lower_bdd}. 
Because $\varepsilon > 0$ is arbitrary, we obtain the desired assertion. 
\end{proof}

\subsection{Proof of Theorem \ref {th_dependent}}

First, note that condition [C5] immediately implies [C2] with $x_0 = 0$ and 
\begin{eqnarray*}
f_X(x) = \int ^\infty _0f_X(x | Y = y)F_Y(dy). 
\end{eqnarray*}
Second, note that by [C6], Proposition 3.1(i) in \cite {Kato_IJTAF} (see also Remark 3.2 therein) 
and Proposition A3.8 in \cite {Embrechts-Klueppelberg-Mikosch}, 
we have (\ref {mdt_Z}) and 
\begin{eqnarray}\label{asy_eq_fX_fZ}
f_X(x) \sim f_Z(x), \ \ x\rightarrow \infty . 
\end{eqnarray}

To prove Theorem \ref {th_dependent}, we give the following three propositions. 

\begin{proposition}\label{prop_VaR_diff}
$\mathrm {VaR}_\alpha (X + uY)$ is continuously differentiable in $u\in [0, 1]$ and it holds that 
\begin{eqnarray*}
\frac{\partial }{\partial u}\mathrm {VaR}_\alpha (X + uY) = 
\E [Y | X + uY = \mathrm {VaR}_\alpha (X + uY)], \ \ 0\leq u\leq 1. 
\end{eqnarray*}
\end{proposition}

Proposition \ref {prop_VaR_diff} is obtained by an argument similar to the proof of 
Lemma 5.3 in \cite {Tasche2000}, using the implicit function theorem. 

\begin{proposition}\label{prop_rv_Euler_VaR}
The function $x\mapsto \E [Y | X = x]$ is regularly varying with index $\kappa + \beta + 1$. 
\end{proposition}

\begin{proof}Fix any $t > 0$. 
We observe that
\begin{eqnarray*}
&&\left| \E [Y | X = tx] - t^{\kappa + \beta + 1}\E [Y | X = x] \right| \\
&\leq & 
\int ^\infty _0y\frac{f_X(x | Y = y)}{f_X(x)}
\left| \frac{f_X(tx | Y = y)}{f_X(x | Y = y)}\cdot \frac{f_X(x)}{f_X(tx)} - t^{\kappa + \beta + 1} \right| F_Y(dy),\\
&\leq & 
\left\{ \sup _{x\geq 0}\frac{f_X(x)}{f_X(tx)}\sup _{y\geq 0}\left| \frac{f_X(tx | Y = y)}{f_X(x | Y = y)} - t^\kappa \right| + 
t^\kappa \left| \frac{f_X(x)}{f_X(tx)} - t^{\beta + 1}\right| \right\} \E [Y | X = x],
\end{eqnarray*}
and therefore, using [C5], we arrive at 
\begin{eqnarray*}
\lim _{x\rightarrow \infty }\left| \frac{\E [Y | X = tx]}{\E [Y | X = x]} - t^{\kappa + \beta + 1} \right| = 0.
\end{eqnarray*}
\end{proof}

\begin{proposition}\label{prop_dependent}
$\E [Y | X = x] \sim 
\E [Y | X + Y = x], \ \ x \rightarrow \infty $. 
\end{proposition}

\begin{proof}
Fix any $\varepsilon > 0$. 
Then we have 
\begin{eqnarray*}
\left| \E [Y | Z = x]  - \E [Y | X = x]\right| 
\leq A^\varepsilon (x) + B^\varepsilon (x), 
\end{eqnarray*}
where we denote $Z = X + Y$ and 
\begin{eqnarray*}
A^\varepsilon (x) &=& \int ^{\varepsilon x}_0y
\left| \frac{f_X(x - y | Y = y)}{f_Z(x)} - \frac{f_X(x | Y = y)}{f_X(x)} \right| F_Y(dy), \\
B^\varepsilon (x) &=& \E [Y1_{\{ Y > \varepsilon x\} } | Z = x] + \E [Y1_{\{ Y > \varepsilon x\} } | X = x]. 
\end{eqnarray*}
By [C7] and the Chebyshev inequality, we get 
\begin{eqnarray}\label{temp_B}
0\leq \frac{B^\varepsilon (x)}{\E [Y | X = x]}\leq 
\frac{C}{\varepsilon ^{\eta - 1}x^{\eta - 1}\E [Y | X = x]},
\end{eqnarray}
for some $C > 0$. 
Because Proposition \ref {prop_rv_Euler_VaR} tells us that 
$x\mapsto x^{\eta - 1}\E [Y | X = x]$ is regularly varying with index $\eta + \kappa + \beta > 0$, 
the right-hand side of (\ref {temp_B}) converges to zero as $x\rightarrow \infty $ 
(see Proposition 1.5.1 in \cite {Bingham-Goldie-Teugels}). 

Moreover, we see that 
\begin{eqnarray*}
A^\varepsilon (x) 
&\leq & 
\int ^{\varepsilon x}_0y\frac{f_X(x | Y = y)}{f_X(x)}
\left| \frac{f_X(x - y | Y = y)}{f_X(x | Y = y)}\cdot 
\frac{f_X(x)}{f_Z(x)} - 1 \right| F_Y(dy)\\
&=& 
\int ^\varepsilon _0ux
\frac{f_X(x | Y = ux)}{f_X(x)}
\left| \frac{f_X((1-u)x | Y = ux)}{f_X(x | Y = ux)}\cdot 
\frac{f_X(x)}{f_Z(x)} - 1 \right| F_{Y/x}(du), \ \ x > 0. 
\end{eqnarray*}
Here, we observe that
\begin{eqnarray*}
&&\left| \frac{f_X((1-u)x | Y = ux)}{f_X(x | Y = ux)}\cdot 
\frac{f_X(x)}{f_Z(x)} - 1 \right| \\
&\leq & 
|1 - (1-u)^\kappa | + (1-u)^\kappa \left| \frac{f_X(x)}{f_Z(x)} - 1\right| + 
\left| \frac{f_X((1-u)x | Y = ux)}{f_X(x | Y = ux)} - (1-u)^\kappa  \right| \frac{f_X(x)}{f_Z(x)}, 
\end{eqnarray*}
for each $u\in [0, \varepsilon ]$. 
Note that if $\kappa \geq 0$ (resp., $\kappa < 0$), we have 
$(1-\varepsilon )^\kappa \leq (1-u)^\kappa \leq 1$ 
(resp., $1\leq (1-u)^\kappa \leq (1-\varepsilon )^\kappa $). 
Moreover, by (\ref {asy_eq_fX_fZ}), 
$f_X(x)/f_Z(x)$ converges to $1$ as $x\rightarrow \infty $,
and so it is bounded. 
Therefore, we get 
\begin{eqnarray*}
A^\varepsilon (x) 
&\leq &
\E [Y | X = x]\Bigg\{ 
|1 - (1 - \varepsilon )^\kappa | + 
\max \{ 1, (1 - \varepsilon )^\kappa \}\left| \frac{f_X(x)}{f_Z(x)} - 1\right| \\
&&\hspace{25mm} + 
C'\sup _{1-\varepsilon \leq t\leq 1}\sup _{y\geq 0}\left| \frac{f_X(tx | Y = y)}{f_X(x | Y = y)} - t^\kappa  \right| \Bigg\}, 
\end{eqnarray*}
for some $C' > 0$. 

Now we arrive at
\begin{eqnarray*}
\limsup _{x\rightarrow \infty }\left| \frac{\E [Y | Z = x]}{\E [Y | X = x]} - 1\right| \leq 
|1 - (1 - \varepsilon )^\kappa | 
\end{eqnarray*}
by using (\ref {cond_D_unif}) and (\ref {asy_eq_fX_fZ}). 
Because $\varepsilon > 0$ is arbitrary, we obtain the desired assertion. 
\end{proof}

\begin{proof}[Proof of Theorem \ref {th_dependent}]
First, note that Proposition \ref {prop_VaR_diff} guarantees that
\begin{eqnarray*}
\E \hspace{0mm}^{Q^X_\alpha }[Y] = \int ^1_0\E [Y | X = x_u]\phi _\alpha (u)du, \ \ 
\E \hspace{0mm}^{Q^{X+Y}_\alpha }[Y] = \int ^1_0\E [Y | X + Y = z_u]\phi _\alpha (u)du,
\end{eqnarray*}
where $x_u = \mathrm {VaR}_u(X)$ and $z_u = \mathrm {VaR}_u(X + Y)$. 

Then, fix any $\varepsilon > 0$. 
By Propositions \ref {prop_rv_Euler_VaR}--\ref {prop_dependent} and Lemma A.3 in \cite {Kato_IJTAF}, 
we see that 
\begin{eqnarray*}
\E [Y | X = x_\alpha ] \sim \E [Y | X = z_\alpha ] \sim \E [Y | X + Y = z_\alpha ], \ \ \alpha \rightarrow 1, 
\end{eqnarray*}
Thus, there is an $\alpha _0 \in (0, 1)$ such that 
\begin{eqnarray*}
\left| \frac{\E [Y | X + Y = z_\alpha ]}{\E [Y | X = x_\alpha ]} - 1 \right| < \varepsilon , \ \ \alpha \in [\alpha _0, 1). 
\end{eqnarray*}
Therefore, we have 
\begin{eqnarray*}
\left| \frac{\E \hspace{0mm}^{Q^{X+Y}_\alpha }[Y]}{\E \hspace{0mm}^{Q^X_\alpha }[Y]} - 1\right|
&\leq & 
\frac{1}{\E \hspace{0mm}^{Q^X_\alpha }[Y]}\int ^1_0
\left| \E [Y | X + Y = z_u] - \E [Y | X = x_u]\right| \phi _\alpha (u)du\\
&\leq & 
\frac{1}{\E \hspace{0mm}^{Q^X_\alpha }[Y]}\int ^{\alpha _0}_0
\left\{  \E [Y | X + Y = z_u] + \E [Y | X = x_u]\right\}  \phi _\alpha (u)du + \varepsilon \\
&\longrightarrow & \varepsilon , \ \ \alpha \rightarrow 1
\end{eqnarray*}
by virtue of (\ref {cond_liminf}). 
Because $\varepsilon > 0$ is arbitrary, we get that $\E \hspace{0mm}^{Q^X_\alpha }[Y]\sim \E \hspace{0mm}^{Q^{X + Y}_\alpha }[Y]$, $\alpha \rightarrow 1$. 
Combining this result with (\ref {ineq_general2}), we obtain the desired assertion. 
\end{proof}

\newpage

\begin{figure}[!h]
\centerline{\includegraphics[height = 8.5cm,width=12cm]{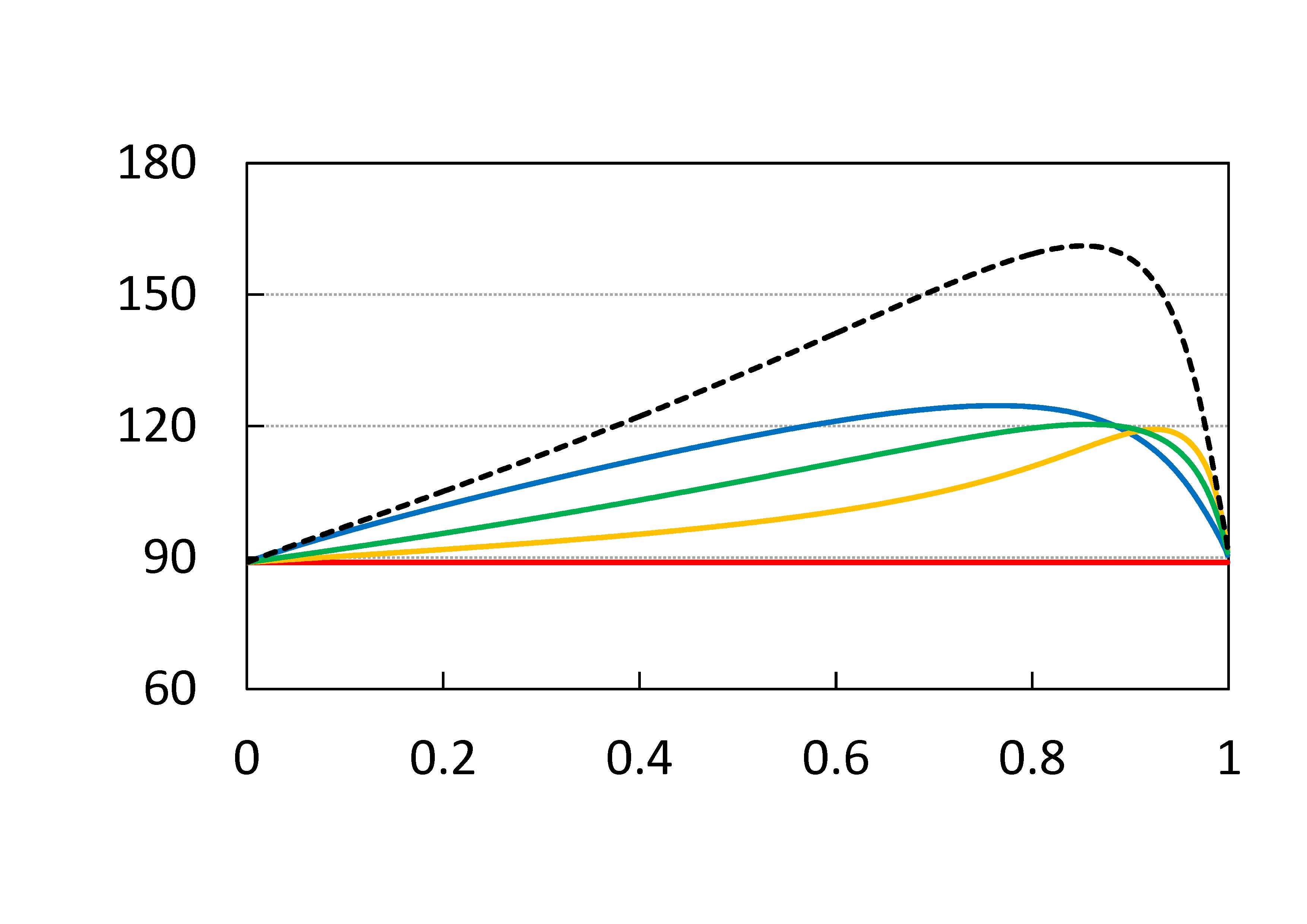}}
\caption{Graphs of $\Delta \mathrm {ES}_\alpha ^{X, Y}$ (blue), 
$\Delta \rho ^{\mathrm {EXP}, X, Y}_\alpha $ (orange), 
$\Delta \rho ^{\mathrm {POW}, X, Y}_\alpha $ (green) and 
$\mathrm {ES}^\mathrm {Euler}_\alpha (Y | X + Y)$ (black, dashed) with $\xi _X = 0.5$ and $\xi _Y = 0.1$. 
The red solid line shows $\E [Y]$.  
The horizontal axis corresponds to $\alpha $. }
\label{fig_1_1}
\end{figure}

\begin{figure}[!h]
\centerline{\includegraphics[height = 8.5cm,width=12cm]{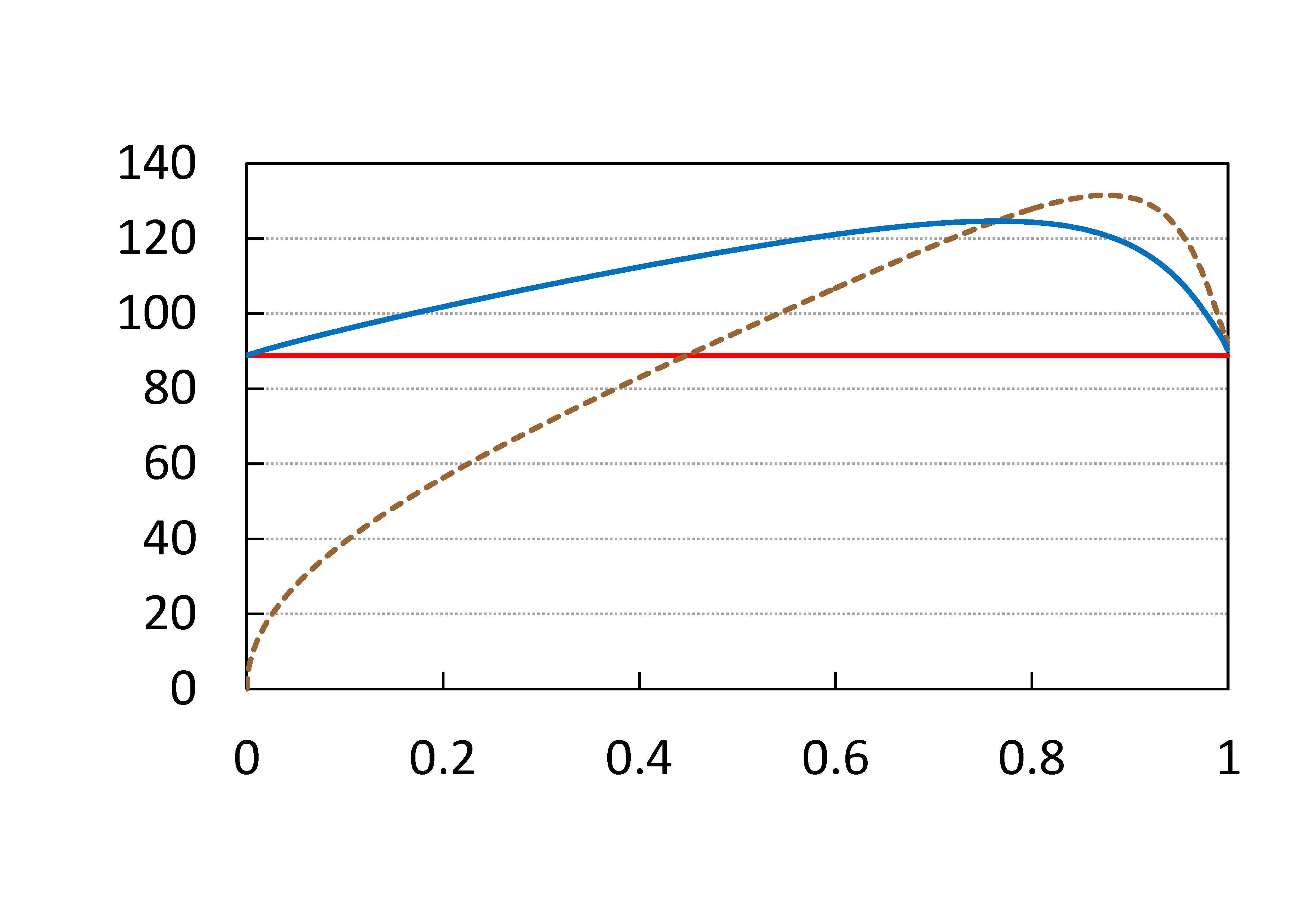}}
\caption{Graphs of $\Delta \mathrm {ES}_\alpha ^{X, Y}$ (blue) and 
$\Delta \mathrm {VaR}^{X, Y}_\alpha $ (brown, dashed) with $\xi _X = 0.5$ and $\xi _Y = 0.1$. 
The red solid line shows $\E [Y]$. 
The horizontal axis corresponds to $\alpha $. }
\label{fig_1_2}
\end{figure}

\begin{figure}[!h]
\centerline{\includegraphics[height = 8.5cm,width=12cm]{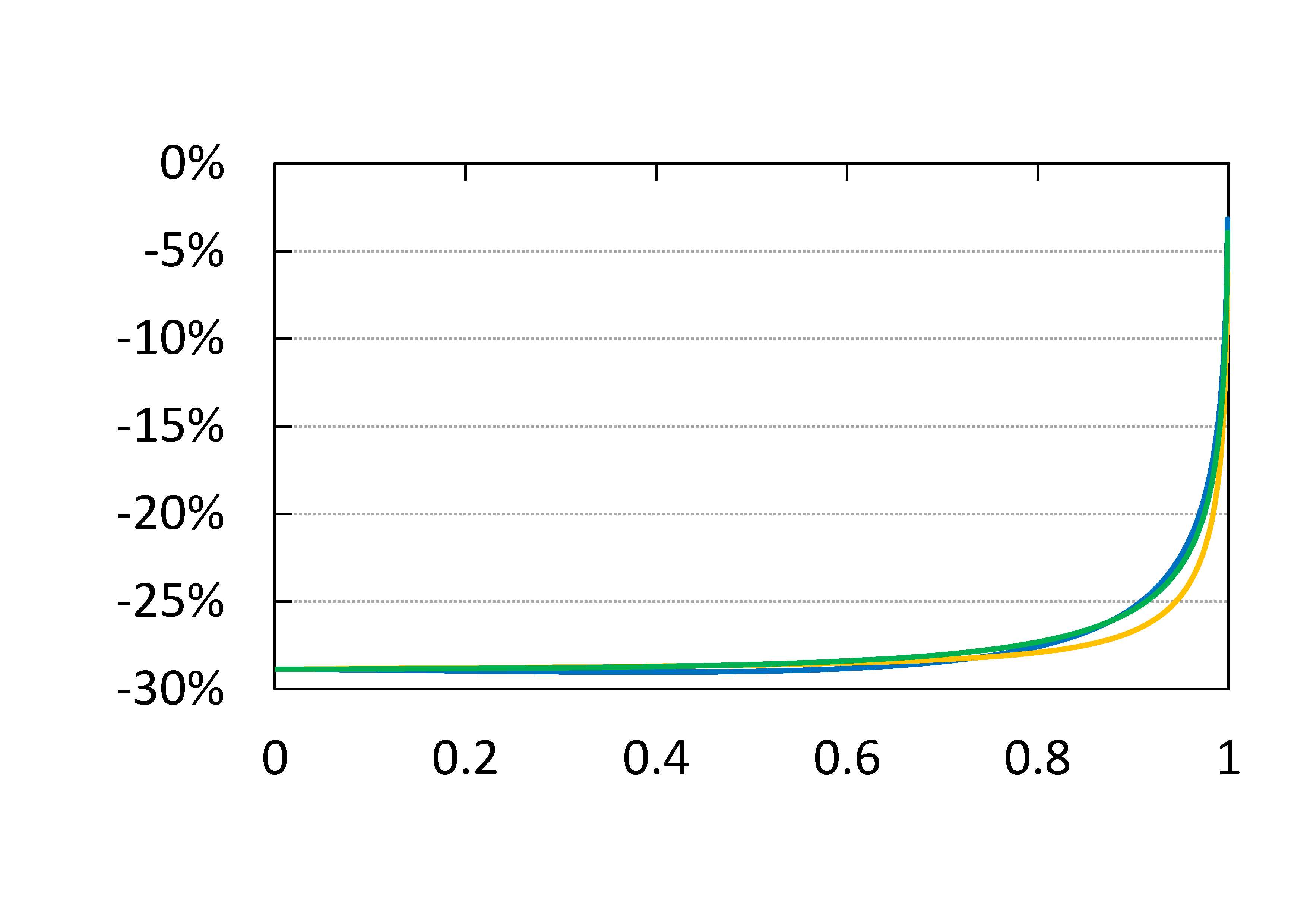}}
\caption{Approximation errors defined by (\ref {def_error}) with $\xi _X = 2/3$ and $\xi _Y = 0.5$ . 
Blue line: $\rho _\alpha = \mathrm {ES}_\alpha $. 
Orange line: $\rho _\alpha = \rho ^\mathrm {EXP}_\alpha $. 
Green line: $\rho _\alpha = \rho ^\mathrm {POW}_\alpha $. 
The horizontal axis corresponds to $\alpha $. }
\label{fig_2_1}
\end{figure}

\begin{figure}[!h]
\centerline{\includegraphics[height = 8.5cm,width=12cm]{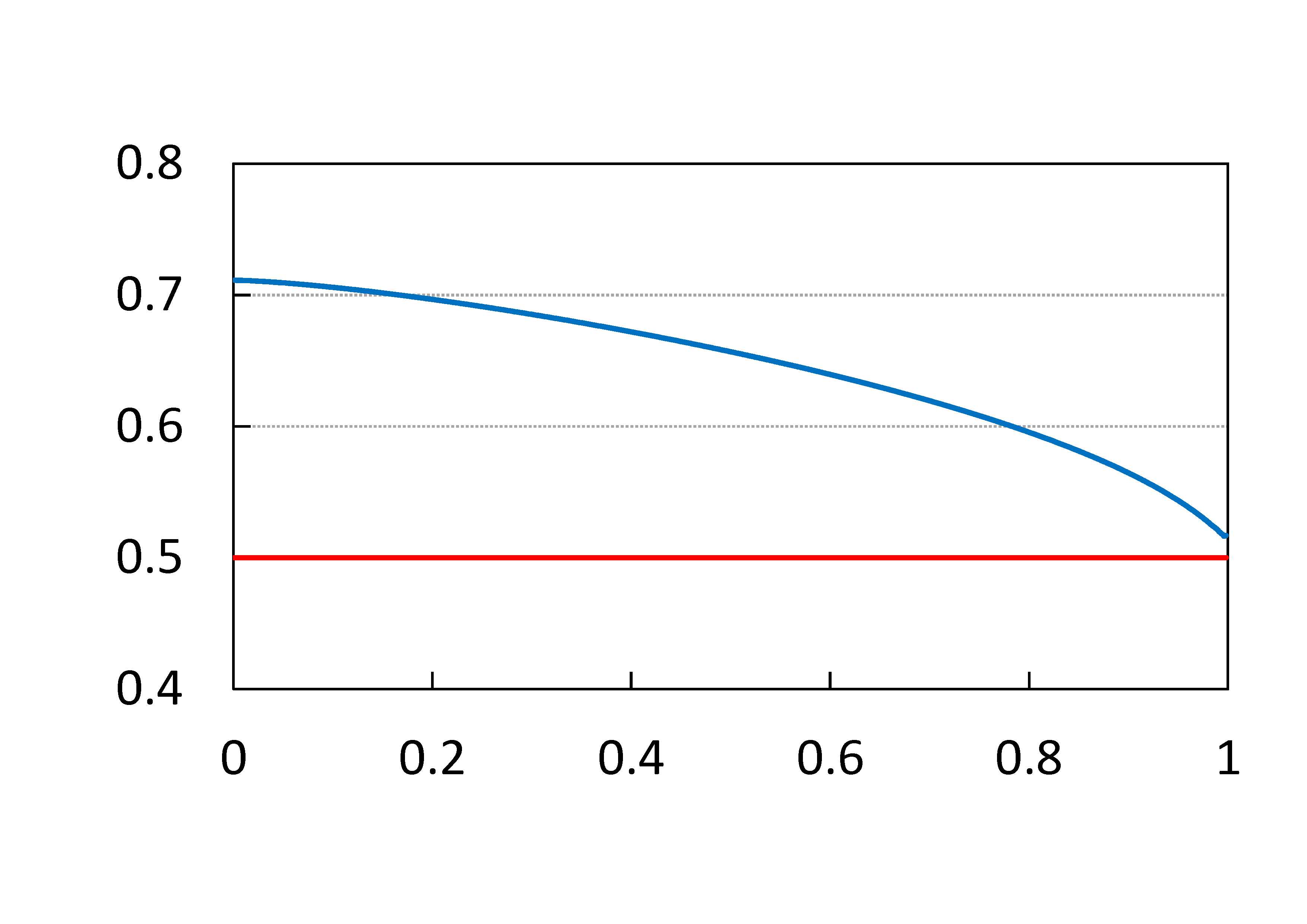}}
\caption{$\bar{M}(\alpha ) / \mathrm {ES}^\mathrm {Euler}_\alpha $ (blue) and 
$\delta = 1 / \gamma = \xi _Y$ (red). 
We set $\xi _X = 2/3$ and $\xi _Y = 0.5$. 
The horizontal axis corresponds to $\alpha $. }
\label{fig_2_1_1}
\end{figure}

\begin{figure}[!h]
\centerline{\includegraphics[height = 8.5cm,width=12cm]{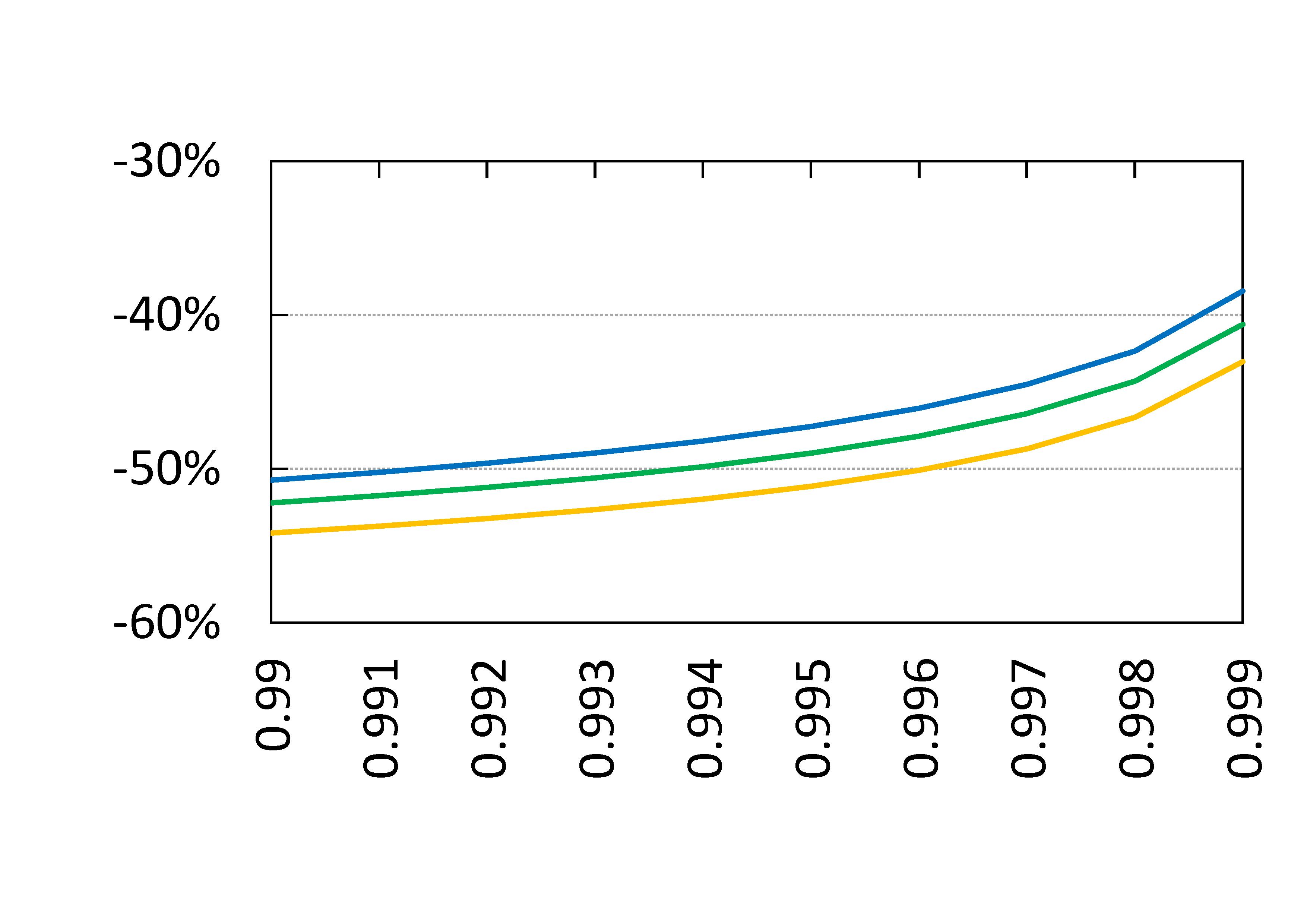}}
\caption{Approximation errors defined by (\ref {def_error}) with $\xi _X = 2/7$ and $\xi _Y = 0.25$. 
Blue line: $\rho _\alpha = \mathrm {ES}_\alpha $. 
Orange line: $\rho _\alpha = \rho ^\mathrm {EXP}_\alpha $. 
Green line: $\rho _\alpha = \rho ^\mathrm {POW}_\alpha $. 
The horizontal axis corresponds to $\alpha $. }
\label{fig_2_2}
\end{figure}

\begin{figure}[!h]
\centerline{\includegraphics[height = 8.5cm,width=12cm]{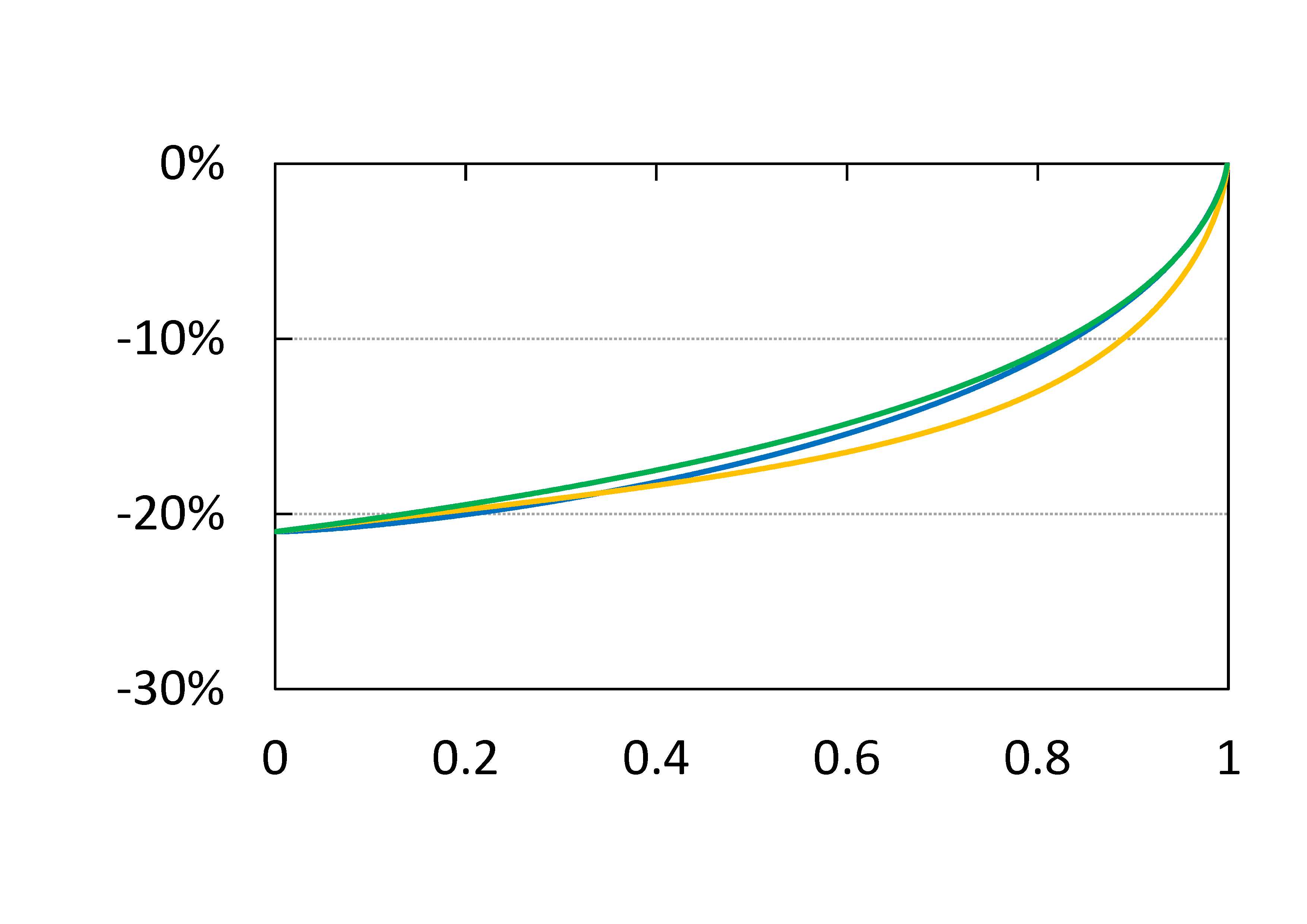}}
\caption{Approximation errors defined by (\ref {def_error}) with $\xi _X = \xi _Y = 0.7$. 
Blue line: $\rho _\alpha = \mathrm {ES}_\alpha $. 
Orange line: $\rho _\alpha = \rho ^\mathrm {EXP}_\alpha $. 
Green line: $\rho _\alpha = \rho ^\mathrm {POW}_\alpha $. 
The horizontal axis corresponds to $\alpha $. }
\label{fig_3_1}
\end{figure}

\begin{figure}[!h]
\centerline{\includegraphics[height = 8.5cm,width=12cm]{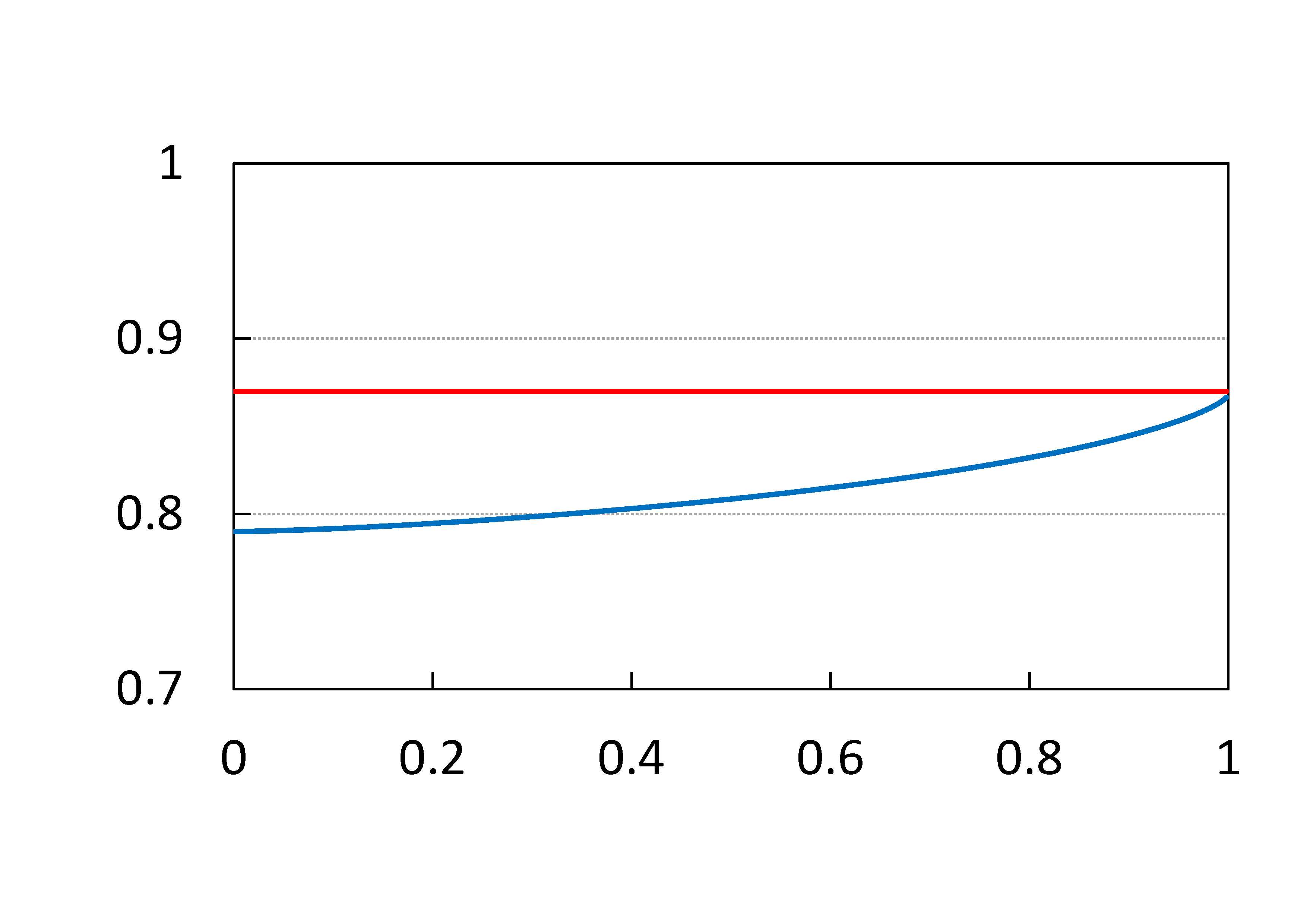}}
\caption{$\bar{M}(\alpha ) / \mathrm {ES}^\mathrm {Euler}_\alpha $ (blue) and 
$\delta = \{1 + k - (1+k)^{1-1/\beta }\} / k$ (red). 
We set $\xi _X = \xi _Y = 0.7$. 
The horizontal axis corresponds to $\alpha $. }
\label{fig_3_1_1}
\end{figure}

\begin{figure}[!h]
\centerline{\includegraphics[height = 8.5cm,width=12cm]{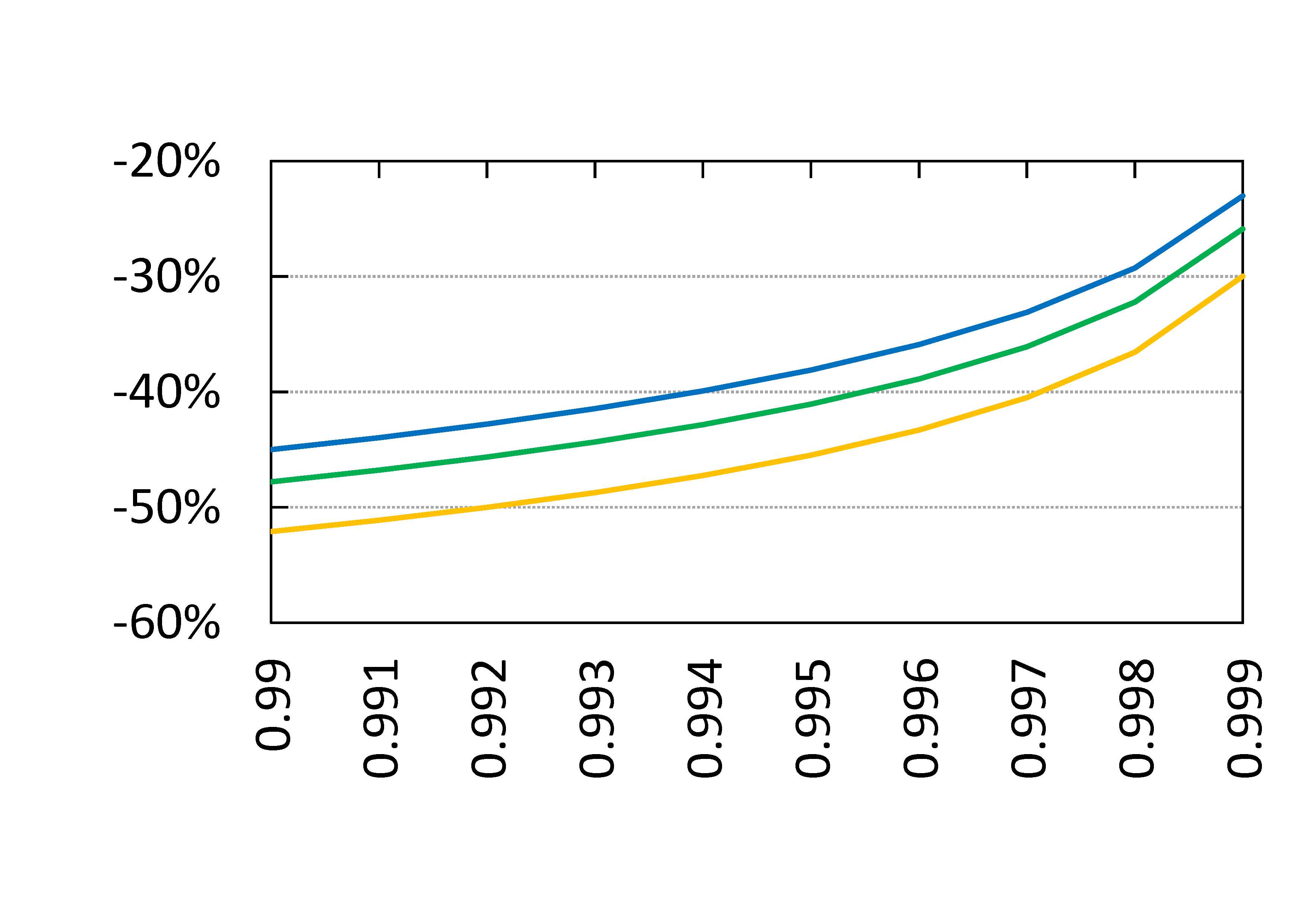}}
\caption{Approximation errors defined by (\ref {def_error}) with $\xi _X = \xi _Y = 0.3$. 
Blue line: $\rho _\alpha = \mathrm {ES}_\alpha $. 
Orange line: $\rho _\alpha = \rho ^\mathrm {EXP}_\alpha $. 
Green line: $\rho _\alpha = \rho ^\mathrm {POW}_\alpha $. 
The horizontal axis corresponds to $\alpha $. }
\label{fig_3_2}
\end{figure}

\begin{figure}[!h]
\centerline{\includegraphics[height = 4.25cm,width=6cm]{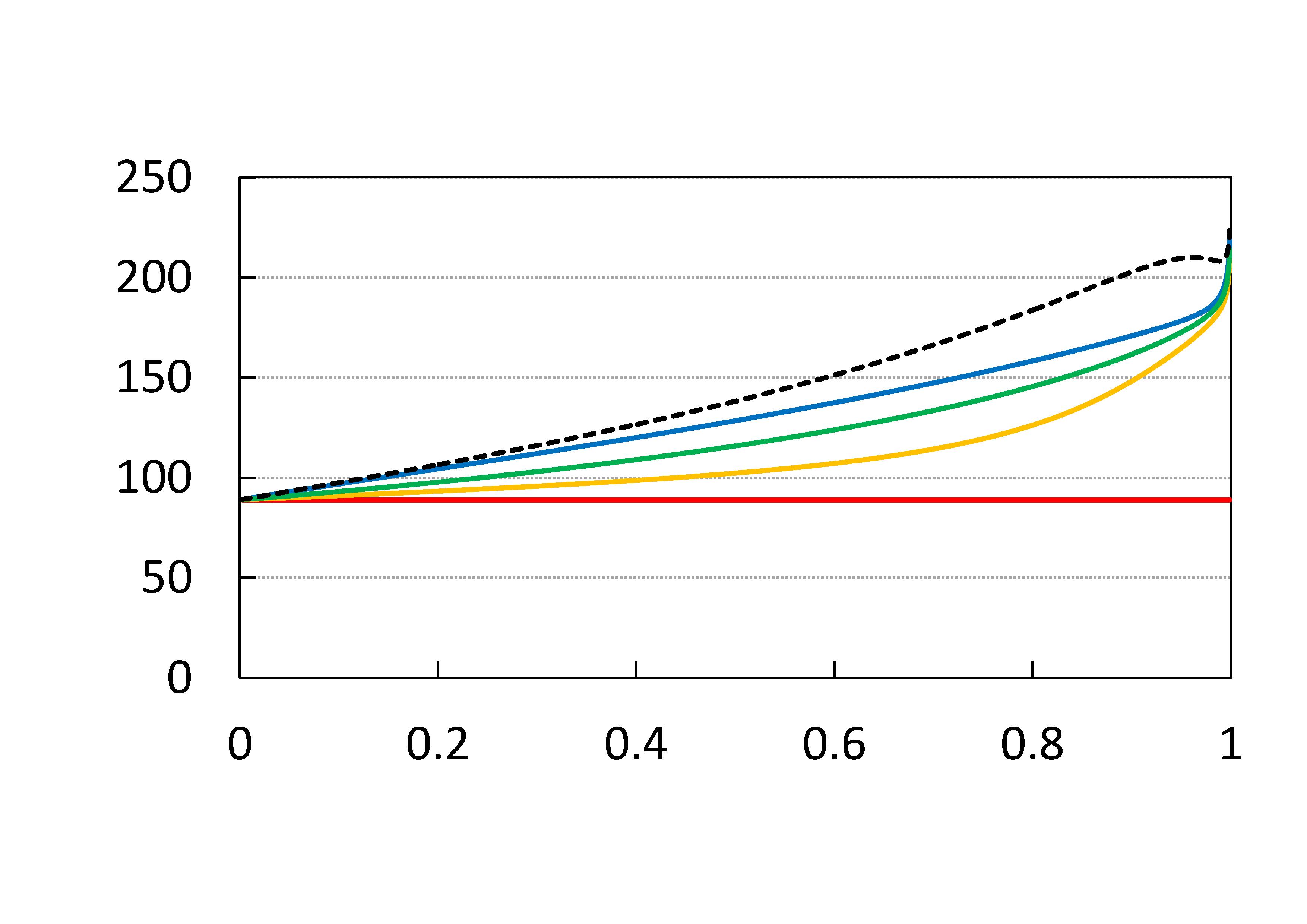}
\includegraphics[height = 4.25cm,width=6cm]{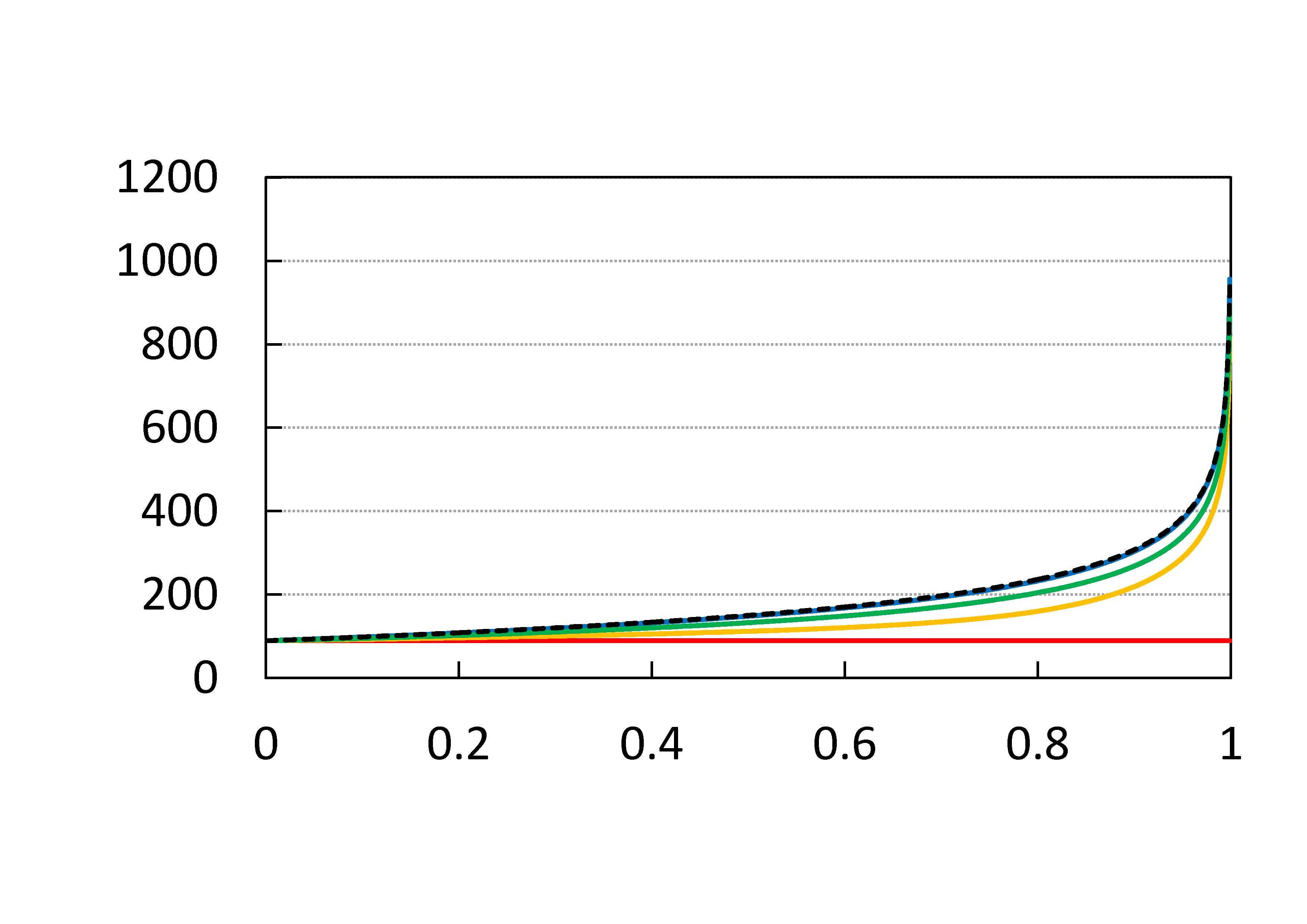}
\includegraphics[height = 4.25cm,width=6cm]{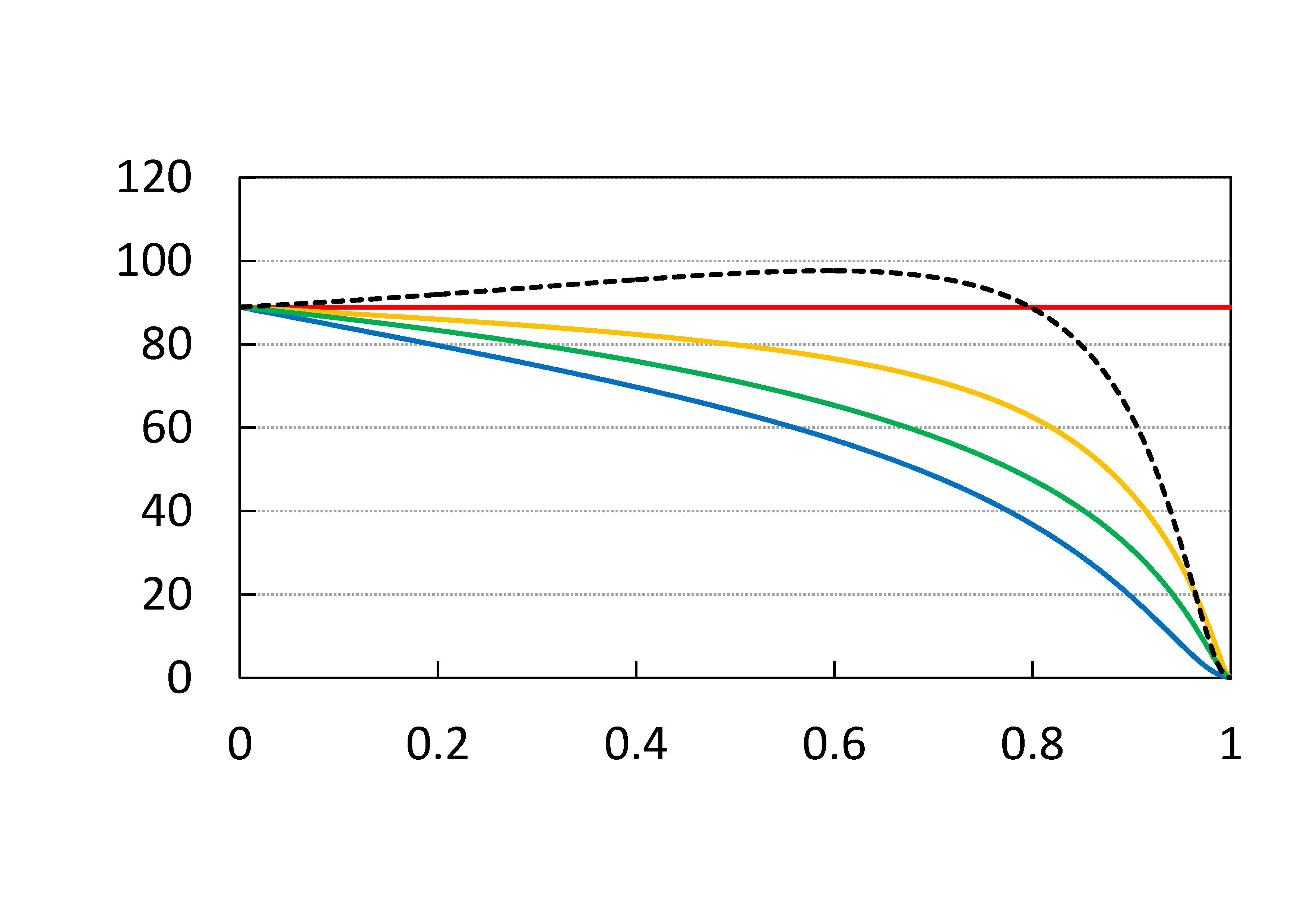}
}
\caption{Graphs of $\Delta \mathrm {ES}_\alpha ^{X, Y}$ (blue), 
$\Delta \rho ^{\mathrm {EXP}, X, Y}_\alpha $ (orange), 
$\Delta \rho ^{\mathrm {POW}, X, Y}_\alpha $ (green) and 
$\mathrm {ES}^\mathrm {Euler}_\alpha (Y | X + Y)$ (black, dashed) with $\xi _X = 0.5$ and $\xi _Y = 0.1$. 
The red solid line shows $\E [Y]$.  
The horizontal axis corresponds to $\alpha $. 
Left: $C = C^\mathrm {Gauss}_\rho $ with $\rho = 0.3$. 
Center: $C = C^\mathrm {Gumbel}_\theta $ with $\theta = 3$. 
Right: $C = C^\mathrm {cmon}$. 
}
\label{fig_A_1}
\end{figure}

\begin{figure}[!h]
\centerline{\includegraphics[height = 4.25cm,width=6cm]{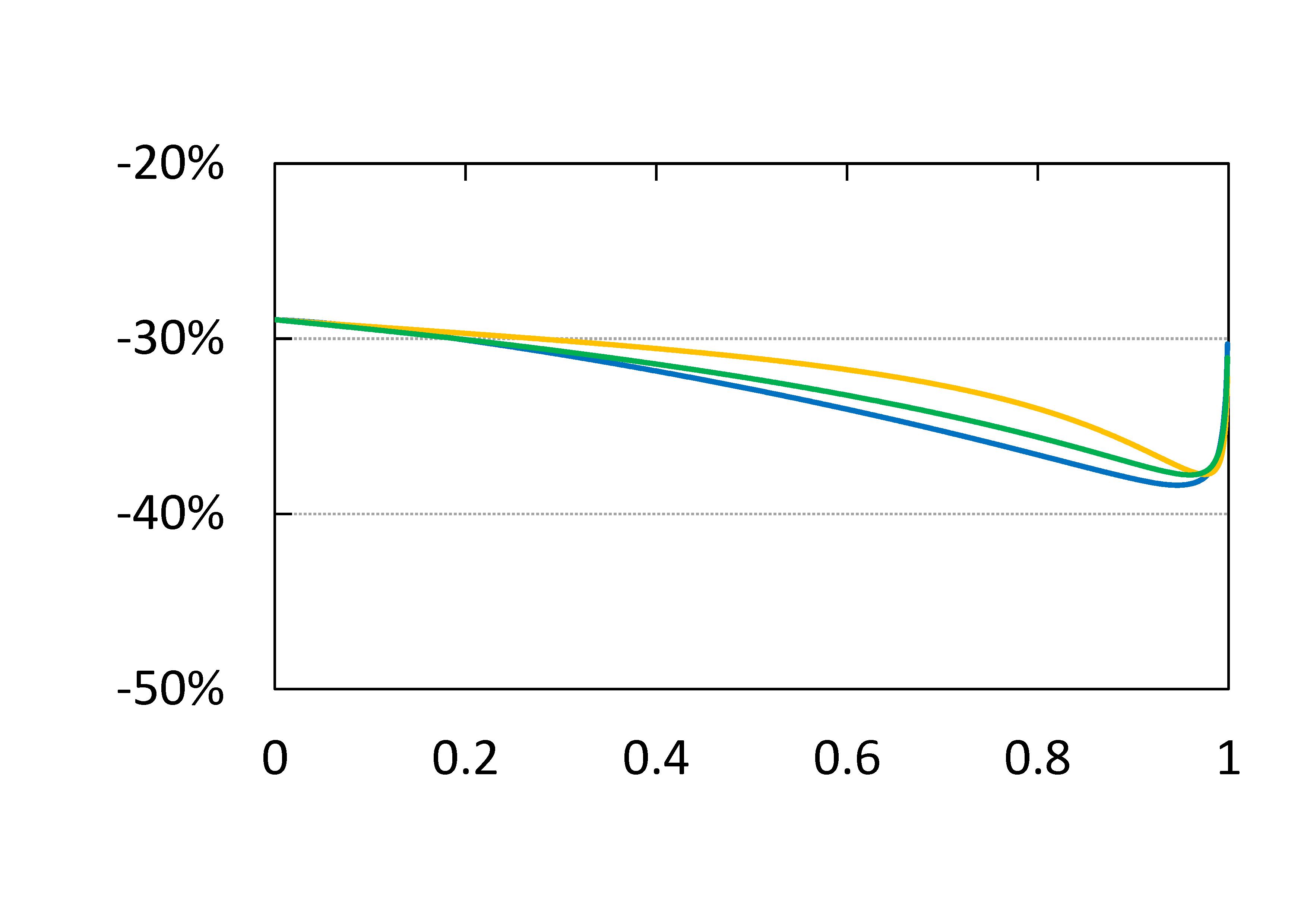}
\includegraphics[height = 4.25cm,width=6cm]{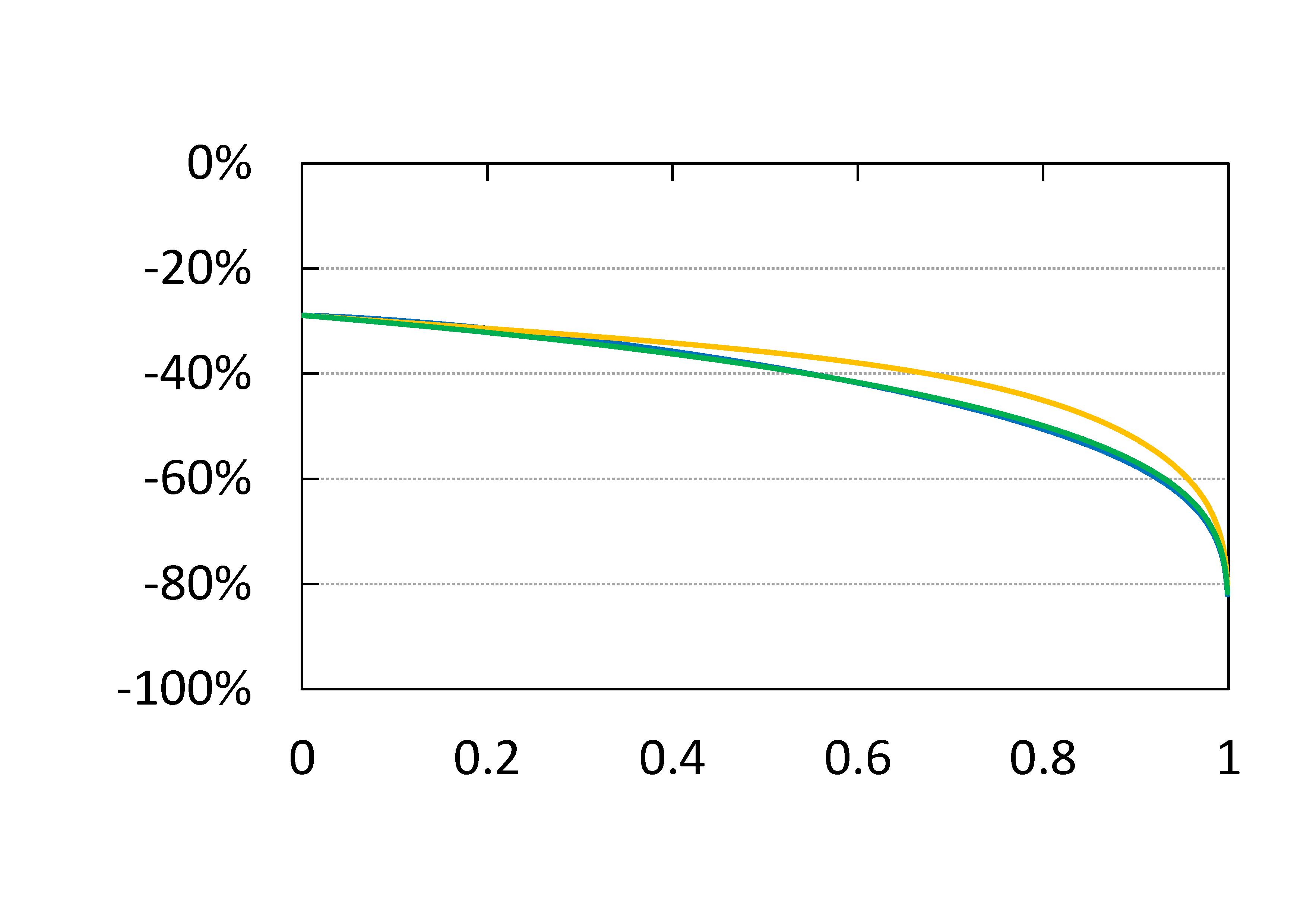}
\includegraphics[height = 4.25cm,width=6cm]{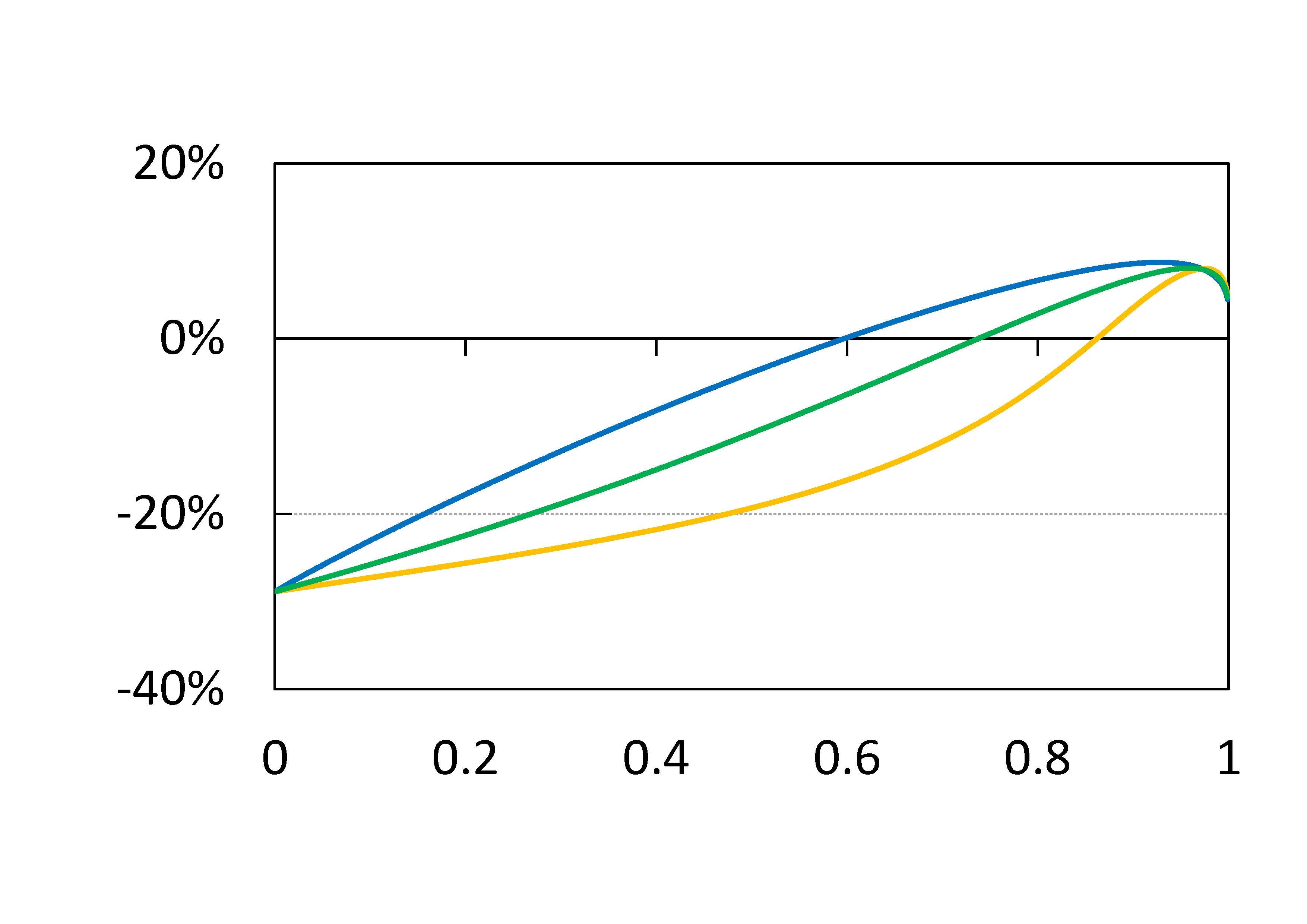}
}
\caption{Approximation errors defined by (\ref {def_error}) with $\xi _X = 2/3$ and $\xi _Y = 0.5$. 
Blue line: $\rho _\alpha = \mathrm {ES}_\alpha $. 
Orange line: $\rho _\alpha = \rho ^\mathrm {EXP}_\alpha $. 
Green line: $\rho _\alpha = \rho ^\mathrm {POW}_\alpha $. 
The horizontal axis corresponds to $\alpha $. 
Left: $C = C^\mathrm {Gauss}_\rho $ with $\rho = 0.3$. 
Center: $C = C^\mathrm {Gumbel}_\theta $ with $\theta = 3$. 
Right: $C = C^\mathrm {cmon}$. 
}
\label{fig_A_2}
\end{figure}

\begin{figure}[!h]
\centerline{\includegraphics[height = 4.25cm,width=6cm]{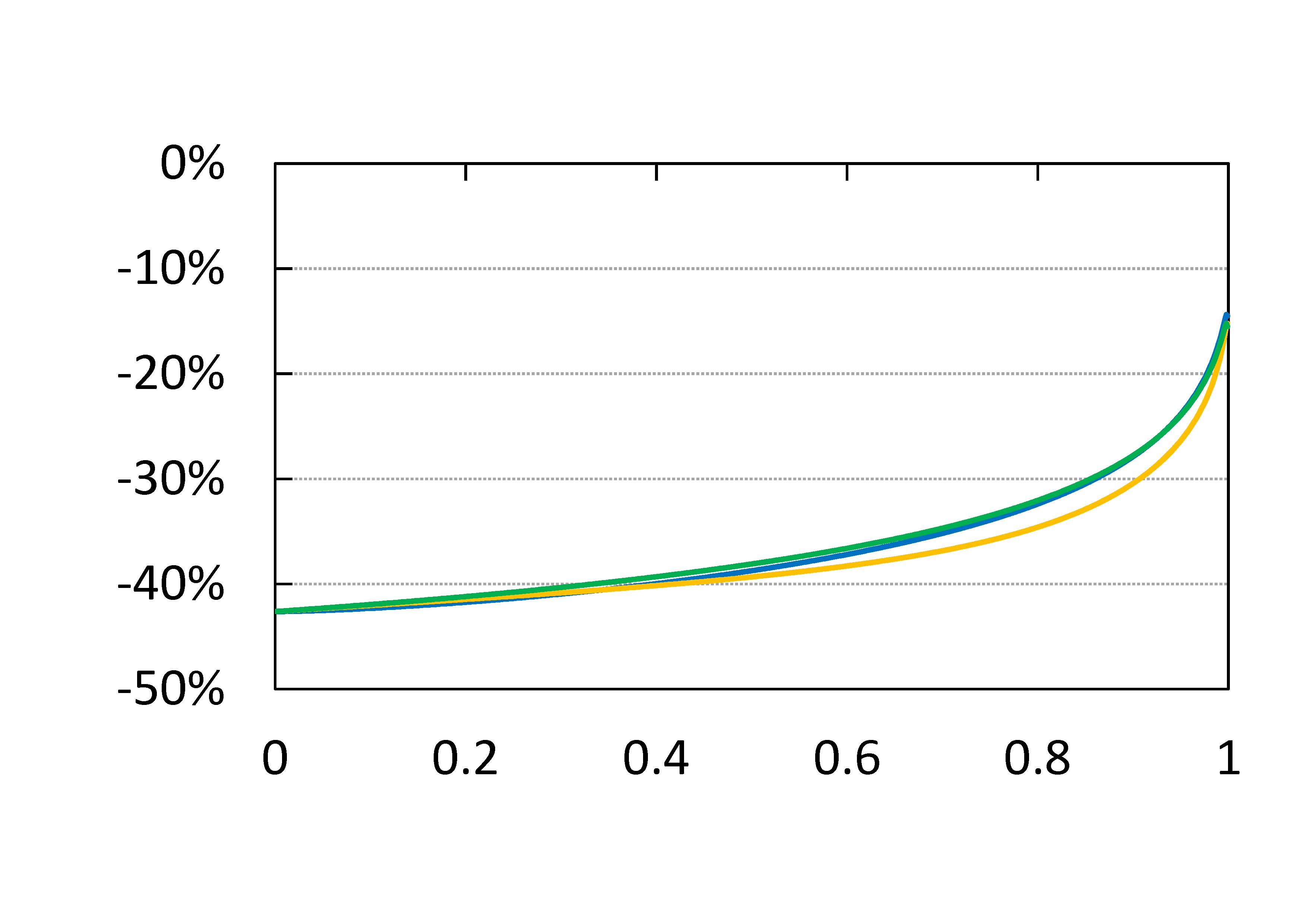}
\includegraphics[height = 4.25cm,width=6cm]{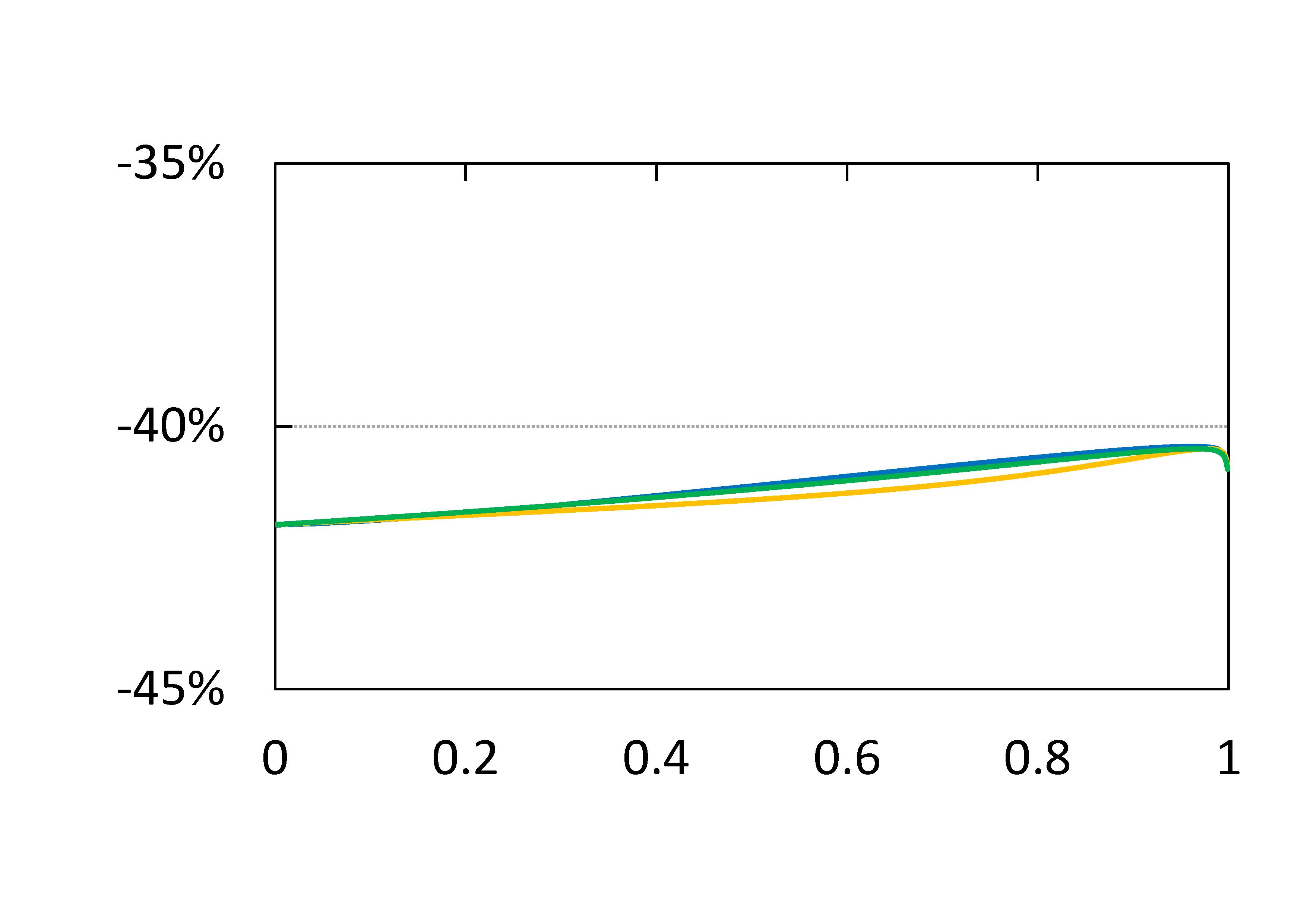}
\includegraphics[height = 4.25cm,width=6cm]{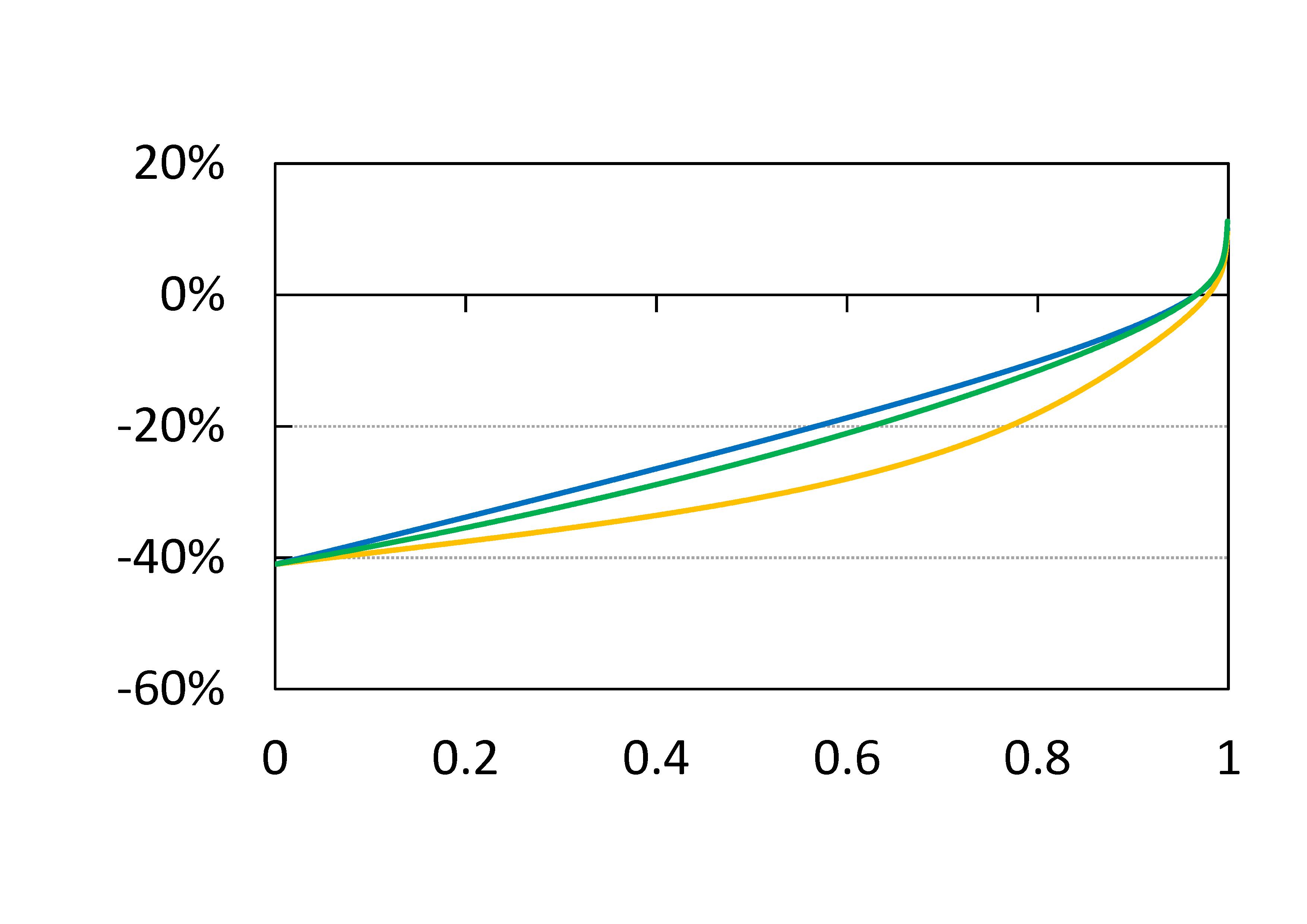}
}
\caption{Approximation errors defined by (\ref {def_error}) with $\xi _X = \xi _Y = 0.7$. 
Blue line: $\rho _\alpha = \mathrm {ES}_\alpha $. 
Orange line: $\rho _\alpha = \rho ^\mathrm {EXP}_\alpha $. 
Green line: $\rho _\alpha = \rho ^\mathrm {POW}_\alpha $. 
The horizontal axis corresponds to $\alpha $. 
Left: $C = C^\mathrm {Gauss}_\rho $ with $\rho = 0.3$. 
Center: $C = C^\mathrm {Gumbel}_\theta $ with $\theta = 3$. 
Right: $C = C^\mathrm {cmon}$. 
}
\label{fig_A_3}
\end{figure}


\begin{thebibliography}{0}


\bibitem 
{Andersson-et-al}
Andersson, F., Mausser, H., Rosen, D. and Uryasev, S. (2001) 
Credit Risk Optimization with Conditional Value-at-Risk Criterion, 
Mathematical Programming Series B, 
\textbf {89}(2), 
273--291. 
\url{https://doi.org/10.1007/PL00011399}



\bibitem 
{Acerbi}
Acerbi, C. (2002) 
Spectral Measures of Risk: A Coherent Representation of Subjective Risk Aversion, 
Journal of Banking \& Finance, 
\textbf {26}(7), 
1505--1518. 
\url{https://doi.org/10.1016/S0378-4266(02)00281-9}


\bibitem 
{Acerbi-Nordio-Sirtori}
Acerbi, C., Nordio, C. and Sirtori, C. (2001) 
Expected Shortfall as a Tool for Financial Risk Management, 
Preprint. 
\url{https://arxiv.org/pdf/cond-mat/0102304.pdf}


\bibitem 
{Acerbi-Tasche1}
Acerbi, C. and Tasche, D. (2002) 
Expected Shortfall: A Natural Coherent Alternative to Value at Risk, 
Economic Notes, 
\textbf {31}(2), 
379--388. 
\url{https://doi.org/10.1111/1468-0300.00091}


\bibitem 
{Acerbi-Tasche2}
Acerbi, C. and Tasche, D. (2002) 
On the Coherence of Expected Shortfall, 
Journal of Banking \& Finance, 
\textbf {26}(7), 
1487--1503. 
\url{https://doi.org/10.1016/S0378-4266(02)00283-2}


\bibitem 
{Artzner-Delbaen-Eber-Heath}
Artzner, P., Delbaen, F., Eber, J.-M. and Heath, D. (1999) 
Coherent Measures of Risk, 
Mathematical Finance, 
\textbf {9}(3), 
203--228. 
\url{https://doi.org/10.1111/1467-9965.00068}


\bibitem 
{BCBS2012}
Basel Committee on Banking Supervision (2012) 
Fundamental Review of the Trading Book, 
Press Release, Bank for International Settlements. 
\url{http://www.bis.org/publ/bcbs219.pdf}


\bibitem 
{BCBS2016}
Basel Committee on Banking Supervision (2016) 
Minimum Capital Requirements for Market Risk, 
Press Release, Bank for International Settlements. 
\url{http://www.bis.org/bcbs/publ/d352.pdf}


\bibitem 
{Bingham-Goldie-Omey}
Bingham, N. H., Goldie, C. M. and Omey, E. (2006) 
Regularly Varying Probability Densities, 
Publications de l'Institut Mathematique, 
\textbf {80}(94), 47--57. 
\url{https://doi.org/10.2298/PIM0694047B}


\bibitem 
{Bingham-Goldie-Teugels}
Bingham, N. H., Goldie, C. M. and Teugels, J. L. (1989) 
Regular Variation, 
Cambridge University Press. 


\bibitem 
{Brandtner-Kuersten}
Brandtner, M. and K\"ursten, W. (2017) 
Consistent Modeling of Risk Averse Behavior with Spectral Risk Measures: W\"achter/Mazzoni Revisited, 
European Journal of Operational Research, 
\textbf {259}(1), 394--399. 
\url{http://dx.doi.org/10.1016/j.ejor.2016.12.027}


\bibitem 
{Cotter-Dowd}
Cotter, J. and Dowd, K. (2006) 
Extreme Spectral Risk Measures: An Application to Futures Clearinghouse Margin Requirements, 
Journal of Banking \& Finance, 
\textbf{30}(12), 3469--3485. 
\url{https://doi.org/10.1016/j.jbankfin.2006.01.008}


\bibitem 
{Dowd-Blake}
Dowd, K. and Blake, D. (2006) 
After VaR: The Theory, Estimation, and Insurance Applications of Quantile-Based Risk Measures, 
The Journal of Risk and Insurance, 
\textbf{73}(2), 193--229. 
\url{http://dx.doi.org/10.1111/j.1539-6975.2006.00171.x}


\bibitem 
{Dowd-Cotter-Sorwar}
Dowd, K., Cotter, J. and Sorwar, G. (2008) 
Spectral Risk Measures: Properties and Limitations, 
Journal of Financial Services Research, 
\textbf{34}(1), 61--75. 
\url{https://doi.org/10.1007/s10693-008-0035-6}


\bibitem 
{Embrechts}
Embrechts, P. (2000) 
Extreme Value Theory: Potential and Limitations as an Integrated Risk Management Tool, 
Derivatives Use, Trading \& Regulation, 
\textbf{6}(1), 449--456. 


\bibitem 
{Embrechts-Klueppelberg-Mikosch}
Embrechts, P., Kl\"uppelberg, C. and Mikosch, T. (1997) 
Modelling Extremal Events, 
Springer. 


\bibitem 
{Foellmer-Schied}
F\"ollmer, H. and Schied, A. (2016) 
Stochastic Finance: An Introduction in Discrete Time 4th rev.~ed., 
De Gruyter. 


\bibitem 
{Jouini-Schachermayer-Touzi}
Jouini, E., Schachermayer, W. and Touzi, N. (2006) 
Law Invariant Risk Measures Have the Fatou Property, 
In: Kusuoka, S. and Yamazaki, A., Ed., 
Advances in Mathematical Economics, 
\textbf{9}, 49--71, Springer, Japan. 
\url{https://doi.org/10.1007/4-431-34342-3_4}


\bibitem 
{Kalkbrener-Kennedy-Popp}
Kalkbrener, M., Kennedy, A. and Popp, M. (2007) 
Efficient Calculation of Expected Shortfall Contributions in Large Credit Portfolios, 
Journal of Computational Finance, 
\textbf {11}(2), 
1--43. 
\url{https://doi.org/10.21314/JCF.2007.162}


\bibitem 
{Kato_IJTAF}
Kato, T. (2017) 
Theoretical Sensitivity Analysis for Quantitative Operational Risk Management, 
International Journal of Theoretical and Applied Finance, 
\textbf {20}(5), 23 pages. 
\url{https://doi.org/10.1142/S0219024917500327}


\bibitem 
{Kusuoka}
Kusuoka, S. (2001) 
On Law-invariant Coherent Risk Measures, 
In: Kusuoka, S. and Maruyama, T., Ed., 
Advances in Mathematical Economics, 
\textbf{3}, 83--95, Springer, Japan. 
\url{https://doi.org/10.1007/978-4-431-67891-5_4}


\bibitem 
{McNeil-Frey-Embrechts}
McNeil, A.J., Frey, R. and Embrechts, P. (2005) 
Quantitative Risk Management, 
Princeton University Press, Princeton. 


\bibitem 
{Pflug-Roemisch}
Pflug, G. Ch. and R\"omisch, W. (2007) 
Modeling, Measuring and Managing Risk, 
World Scientific Publishing Co., London. 


\bibitem 
{Puzanova-Duellmann}
Puzanova, N. and D\"ullmann, K. (2013) 
Systemic Risk Contributions: A Credit Portfolio Approach, 
Journal of Banking \& Finance, 
\textbf{37}(4), 1243--1257. 
\url{https://doi.org/10.1016/j.jbankfin.2012.11.017}


\bibitem 
{Shapiro}
Shapiro, S. (2013) 
On Kusuoka Representation of Law Invariant Risk Measures, 
Mathematics of Operations Research, 
\textbf{38}(1), 142--152. 
\url{https://doi.org/10.1287/moor.1120.0563}


\bibitem 
{Sriboonchitta-Nguyen-Kreinovich}
Sriboonchitta, S., Nguyen, H. T. and Kreinovich, V. (2010) 
How to Relate Spectral Risk Measures and Utilities, 
International Journal of Intelligent Technologies and Applied Statistics, 
\textbf{3}(2), 141--158. 
\url{https://doi.org/10.6148/IJITAS.2010.0302.03}


\bibitem 
{Tasche2000}
Tasche, D. (2000) 
Risk Contributions and Performance Measurement, 
Working Paper. 
\url{https://pdfs.semanticscholar.org/2659/60513755b26ada0b4fb688460e8334a409dd.pdf}


\bibitem 
{Tasche2002}
Tasche, D. (2002) 
Expected Shortfall and Beyond, 
Journal of Banking \& Finance, 
\textbf {26}(7), 
1519--1533. 
\url{https://doi.org/10.1016/S0378-4266(02)00272-8}


\bibitem 
{Tasche2008}
Tasche, D. (2008) 
Capital Allocation to Business Units and Sub-portfolios: The Euler Principle, 
In: Resti, A., Ed., Pillar II in the New Basel Accord: The Challenge of Economic Capital, 
423--453, Risk Books, London. 


\bibitem 
{Waechter-Mazzoni}
W\"achter, H.P. and Mazzoni, T. (2013) 
Consistent Modeling of Risk Averse Behavior with Spectral Risk Measures, 
European Journal of Operational Research, 
\textbf {229}(2), 
487--495. 
\url{https://doi.org/10.1016/j.ejor.2013.03.001}


\end{thebibliography}
\end{document}